%% file: main.tex
\documentclass[%
 reprint,
 showkeys,
 amsmath,amssymb,
 aps,
 pra,
 nofootinbib % show footnotes as footnotes
]{revtex4-2}

\usepackage{graphicx}
\usepackage{dcolumn}
\usepackage{bm}% bold math
\usepackage{color} %use... colors 
\usepackage{lipsum} %lorem ipsum dolor sit amet...
\usepackage{amsmath} % here i have loaded this to deal well with multiple equations
\usepackage{amsthm} % theorem-like environments
\usepackage{natbib}	%nice citation options
\setcitestyle{numbers} %Citation-related commands
\usepackage{braket}
\usepackage{enumitem} %change items in enumerate, for example, to roman
\usepackage{hyperref}% add hypertext capabilities
\hypersetup{
    colorlinks=true,
    urlcolor=blue,
    linkcolor=blue,		% hyperlinks with \ref{} and \eqref{} obey this one
    citecolor=blue}
\newif\ifPacks
\Packstrue

\ifPacks
	\usepackage{bbold} %blackboard bold NUMBERS
	\usepackage{tikz,tikz-cd}
\fi
\newif\ifHighlights
\Highlightsfalse

\ifHighlights
	\newcommand{\changed}[1]{{\color{red}#1}}
\else
	\newcommand{\changed}[1]{#1}
\fi

\begin{document}

%\preprint{APS/123-QED}

%\title{A negative result on effective energy observables for bipartite quantum systems}
%\title{A negative result on effective energy observables for open quantum systems}
\title{A constraint on local definitions of quantum internal energy}
%\title{A constraint on the possible definition of quantum internal energy}
%\title{On effective energy observables for bipartite quantum systems}

\author{Luis Rodrigo Torres Neves}
\email{rodrigoneves@usp.br}
\affiliation{
 Instituto de F\'isica de S\~ao Carlos, Universidade de S\~ao Paulo -- S\~ao Carlos (SP), Brasil 
}
\author{Frederico Brito}
\email{fbb@ifsc.usp.br}
\affiliation{
 Instituto de F\'isica de S\~ao Carlos, Universidade de S\~ao Paulo -- S\~ao Carlos (SP), Brasil 
}
\affiliation{Quantum Research Centre, Technology Innovation Institute, P.O. Box 9639, Abu Dhabi,UAE}

\date{\today}

\input{shortcuts.tex}

\begin{abstract}
Recent advances in quantum thermodynamics have been focusing on ever more elementary systems of interest, approaching the limit of a single qubit, with correlations, strong coupling and far-from-equilibrium environments coming into play. Under such scenarios, it is clear that fundamental physical quantities must be revisited. This article questions whether a universal definition of internal energy for open quantum systems can be devised, setting limits on its possible properties. We argue that, for such a definition to be regarded as local, it should be determined by using only local resources, i.e., the open system's reduced density operator $\varrho$ and its time derivatives. The simplest construction, then, would be a functional $U(\varrho, \dot{\varrho})$. We adopt the minimalist implementation of a bipartite quantum universe, namely two qubits in a pure joint state, and show that the functional relationship cannot be that simple if it is to generally recover the well-established internal energy of the universe. No further hypothesis or approximation scheme was assumed. An illustrative counter-example is explored, and possible implications of the general constraint are discussed.

\end{abstract}

%Use showkeys class option if keyword display desired

\maketitle

%\tableofcontents

\section{Introduction}\label{sec:intro}

	One of the central tasks in the contemporary theoretical field of quantum thermodynamics is investigating to what extent, if any, the principles and concepts of thermodynamics may be applied to far-from-equilibrium quantum mechanical systems, ultimately those containing a single degree of freedom. 
Arguably, the most fundamental challenges are faced under the presence of correlations, strong system-environment coupling, far-from-equilibrium initial conditions, and most especially when every physical entity is explicitly conceived as a quantum mechanical system. Under such a general regime, it is not clear how, or even whether, one should define the internal energy of an open quantum system. 
	
	In the past few years, in an effort to expand the scope of Alicki’s paradigmatic approach \cite{Alicki_1979}, numerous contributions have been partly or totally devoted to the problem of addressing definitions of \textit{work} and/or \textit{heat} for open systems in the presence of one or many of the above mentioned generalizing features \cite{Weimer_etal_2008, Hossein-Nejad_etal_2015, Alipour_etal_2016, Rivas_2020, Silva_Angelo_2021, Colla_Breuer_2021, Alipour_etal_2022}. Explicitly or not, each relies on a particular notion of internal energy. In turn, still on fairly general contexts, other recent works address specific questions in quantum thermodynamics while relying, often implicitly, on presumably valid definitions of open-system internal energy 
\cite{
	Ochoa_etal_2016,  Dou_etal_2018, Valente_etal_2018, Pyharanta_etal_2022, Strasberg_2019,  
	Carrega_etal_2016, Micadei_etal_2019, Ali_Huang_Zhang_2020,Bernardo_2021}. 
Since there is no consensus with regards to what physical principle should guide such a definition, many of those approaches, being by principle applicable to identical physical situations, may be expected to fall into contradiction with one another. 

	In this work, we are concerned with the full quantum paradigm, in which the presence of external (classical) driving is expendable. For simplicity, we investigate a closed, autonomous, bipartite quantum system, \ie{}, a ``quantum universe'' consisting of two coupled, individual quantum systems. Therefore, our main interest is the question of \textit{energy exchange between (two) quantum systems}. In this context, the basic element always considered is the universe Hamiltonian, $H=H^S+\hint+H^E$, with three additive components -- usually labelled system ($S$), interaction ($I$), and environment ($E$), respectively, which are often taken to depend on time, but assumed to be constant in our main discussion. In this realm, some authors assume the internal energy of $S$ to be given by the quantum mechanical average of its ``bare'' Hamiltonian, $U^S=\avg{H^S}$ \cite{Carrega_etal_2016, Micadei_etal_2019, Ali_Huang_Zhang_2020, Bernardo_2021, Alipour_etal_2022}, whereas many others sustain that, in general, the notion of internal energy should rather derive from a corrected energy observable, sometimes called an ``effective Hamiltonian'', incorporating effects of the interaction with the environment \cite{Weimer_etal_2008, Hossein-Nejad_etal_2015, Alipour_etal_2016, Ochoa_etal_2016, Dou_etal_2018, Valente_etal_2018, Rivas_2020, Colla_Breuer_2021, Pyharanta_etal_2022}. Arguably, that typically means that the interaction energy is somehow divided among the two interacting parties, a hypothesis made explicit for instance in \cite{Hossein-Nejad_etal_2015, Ochoa_etal_2016, Dou_etal_2018}; it has also been proposed that the decomposition should include a third share, not assigned to either system but to the correlations between them \citep{Alipour_etal_2016}. 
		
	We place ourselves closer to the effective Hamiltonian point of view, for two reasons. In the first place, it is widely accepted that the physically observable energy of a quantum mechanical system is modified by the presence of an environment, as in the case of the Lamb and Stark shifts.\footnote{\textit{Cf.} \cite{Welton_1948}; see also \cite{Breuer_Petruccione_2002}, pp. 136, 145, and 586.} Accordingly, it is a commonplace in the field of open quantum systems to speak of ``renormalized system frequency'', ``instantaneous system frequency'', ``time-dependent frequency shift'', and so forth, to refer to the Hamiltonian term in the equation of motion, which is typically identified with the Lamb shift in the proper limit \citep{Breuer_Petruccione_2002, Vacchini_Breuer_2010, Rivas_Huelga_2012, de_Vega_Alonso_2017, Zhang_2019, Ali_Huang_Zhang_2020}. It is thus desirable that a general notion of internal energy be able to account for such renormalization. Secondly, a splitting of the non-negligible interaction energy between the two subsystems is of course necessary if one expects to be able to consistently address a separate thermodynamical description of each of them -- with a version of the First Law, for instance --, which is always an assumption, though possibly a tacit one, if one attempts to define internal energy in the first place. 
	
	Therefore, our departure point is to conceive internal energy in the most abstract possible way -- as merely a functional of the ``relevant parameters'' of the system. Usually, such a functional is defined as the average of a certain Hermitian operator instantaneously assigned to the system, often called an ``effective Hamiltonian''. The natural question, for which many different answers have been proposed, is \textit{how} such an observable should be defined \cite{Hossein-Nejad_etal_2015, Alipour_etal_2016, Dou_etal_2018, Valente_etal_2018, Rivas_2020, Colla_Breuer_2021} -- or, more abstractly, which functional \changed{encodes} the internal energy. The purpose of the present article is \textit{not} to introduce an additional definition for internal energy. Instead, we postulate some properties to be satisfied by such a hypothetical, \textit{universal} definition; then, on the basis of the most minimalist implementation of a bipartite quantum universe, we show that they cannot be satisfied simultaneously.
	
	One basic requirement is that a definition of internal energy for open quantum systems should be compatible with the pre-existing notion of internal energy of a \textit{closed} system. In the ``quantum universe'' paradigm, this has one strong implication: the internal energy of the system \textit{and that of the environment} should add up to the average of the total Hamiltonian, $U^S+U^E=\avg{H}$. This property, which from now on we shall simply refer to as \textit{consistency}, is satisfied by construction, for instance, in the approaches of \cite{Ochoa_etal_2016, Dou_etal_2018} 
on the basis of very particular physical models, and also in \cite{Alipour_etal_2016} with the additional correlation term; but is not discussed in \cite{Weimer_etal_2008, Rivas_2020, Colla_Breuer_2021, Valente_etal_2018}. Giving such a consistent definition is equivalent to proposing a systematic way of splitting the average interaction energy $\avg{\hint}$ between the two interacting systems.

%%%%%%%%%%%%%%%%%%%%%%%%%%%%%%%%%%%%%%%%%%%%
%%%%%%%%%%%%%%%%%%%%%%%%%%%%%%%%%%%%%%%%%%%%

%\iffalse BEGINNING OF VERSION TWO
	An elementary question is to determine \textit{what variables} the internal energy of an open quantum system should be taken to depend on. One of the most elegant aspects of classical, equilibrium thermodynamics is that it allows for an \textit{a priori} complex object to be described by a small number of variables, 
representing macroscopically observable quantities \cite{Callen_1985, Fermi_1937}.

Then, the so-called (equilibrium) \textit{state} of the system of interest is taken to be completely characterized by the set of these variables.

In this sense, internal energy is a \textit{state function}, which corresponds to the fact that it is the system's \textit{internal} degrees of freedom that determine its energy content, since the energy exchange with an environment, when it exists, is negligibly small. %moreover, the time derivatives of the parameters defining the state are all known to be zero, by the hypothesis of equilibrium.

In contrast, in quantum thermodynamics, the system of interest is ``simple'' in nature -- it contains few degrees of freedom -- but is subject to non-equilibrium dynamics, under non-negligible energy exchange with its environment, and non-classically correlated to the latter, in general. 
Would the system of interest be an isolated system, quantum mechanics would establish its internal energy in a well-defined way as the expectation value of the system's Hamiltonian ($U:= \avg{H} = \tr{\rho H}$, where $\rho$ is the system's density operator). Worthy of note is the fact that one can cast the internal energy of an isolated quantum system as a functional of solely the density operator and its first time-derivative, \ie{}, $U = U(\rho, \dot{\rho})$ (see Appendix \ref{sec:app:closed_system}). And hence not requiring the knowledge of the Hermitian operator associated with the energy observable if the time-series of $\rho$, and consequently $\dot{\rho}$, is known by any means, \eg{}, state tomography. Here, we take such a result observed in isolated systems to put forward the minimalist hypothesis that the internal energy of an open quantum system should follow (at least) the functional dependence $U = U(\varrho^S, \dot{\varrho}^S)$, where $\varrho^S := \trp{E}\rho$ denotes the system's reduced density operator. Observe that under such a view, even without exactly knowing the system's Hamiltonian, an observer could determine changes in the system's internal energy if $\varrho^S$ and $\dot{\varrho}^S$ are known. It is important to reinforce that the minimalist hypothesis raised does not mean that the internal energy is not dependent on the Hamiltonian. Based on what is found for isolated systems, it conjectures that $U$ could be determined from the resources that can be obtained from the \textit{kinematical} knowledge of $\varrho(t)$. If so, then one would have an operational approach to determine $U$ once $\varrho(t)$ is known. In short: not knowing what variables should the internal energy of an open quantum system depend of, we investigate the minimalist possibility, $U = U(\varrho^S, \dot{\varrho}^S)$.

Such a construction would also imply that, if two hypothetical physical situations are such that $S$ appears the same to a \textit{local} observer that can only perform local state tomography (thus recording only $\varrho^S$ as a function of time, much like in the spirit of the ``operational approach'' of open quantum systems \cite{Milz_etal_2017}), then $S$ would be assigned the same internal energy in both cases. Of course, a general definition of internal energy with this property could additionally depend on higher-order derivatives of $\varrho^S$. That interesting feature is a type of \textit{locality}, which in a broader sense is often sought in definitions of quantum thermodynamical variables (\eg{} \cite{Alipour_etal_2016, Rivas_2020, Colla_Breuer_2021}). We know of at least one proposal \cite{Valente_etal_2018} that fits in this category, furthermore with the ``minimalist'' dependence $U=U(\varrho^S,\dot{\varrho}^S)$, and we make a short account of it in what follows.
%\fi END OF VERSION TWO

   One possible approach to defining internal energy for open quantum systems is to conceive the effective Hamiltonian as the instantaneous generator $\tilde{H}(t)$ of the unitary component of the dynamics when an equation of the form $\dot{\varrho}^S(t) = -i\left[\tilde{H}(t),\varrho^S(t)\right] + \mathcal{D}(t)\left\lbrace\varrho^S(t)\right\rbrace$ on the system's density operator $\varrho^S(t)$ is available, which is known to be the case in fairly general contexts.\footnote{See \eg{} \cite{de_Vega_Alonso_2017}, pp. 17-18, and also \cite{Alipour_etal_2020}.} (We set $\hbar=1$, denote time derivatives with a dot ($\dot{\square}$), and adopt the Schrödinger picture.) An advantage of this procedure, which to our knowledge first appeared in \cite{Weimer_etal_2008}, is that it automatically incorporates the effects of the environment, and would recover the well-known Lamb shift in the weak-coupling, Markovian limit \cite{Breuer_Petruccione_2002, de_Vega_Alonso_2017}. In turn, the property here named \textit{consistency} should be shown to hold \textit{a posteriori}. Moreover, this prescription is doubtful insofar as the decomposition of the dynamical equation into unitary and non-unitary parts is highly non-unique, whence the same physical situation might be assigned different values of the amount of exchanged energy. Part of the freedom in defining the unitary component amounts to the obvious invariance of the commutator term alone: $\tilde{H}(t) \mapsto \tilde{H}(t)+\xi(t)$, the operator $\xi(t)$ commuting with $\varrho^S(t)$; in particular, any $\xi(t) = \alpha(t)\mathbb{1}^S$, $\alpha(t)\in\mathbb{R}$, is allowed. This particular arbitrariness can be regarded as being of classical nature, for it has a clear counterpart in the classical Liouville equation.\footnote{This lack of unicity in the classical realm drove an important debate in stochastic thermodynamics \citep{Horowitz_Jarzynski_2008, Peliti_2008b, Vilar_Rubi_2008b}, the conclusion of which relied on the fact that, in that context, \textit{system} and \textit{environment} played explicitly distinct roles with respect to an external agent \cite{Peliti_2008a}, a distinction that is not present in the context of our interest.} In the quantum case, there exists additional freedom since unitary and dissipative parts can be \textit{jointly} transformed without modifying their sum total.\footnote{This issue recently motivated a full-length study \cite{Hayden_Sorce_2022}, which proposes to solve such ambiguity by postulating a particular kind of extremization principle. Their approach was later employed as the basis of a quantum thermodynamical formalism \cite{Colla_Breuer_2021}, which however does not address the question of consistency, highlighted here.}

   Despite the issues just mentioned, let us briefly consider the approach of \cite{Valente_etal_2018}, in which the instantaneous internal energy of the system of interest $S$ is again defined as the average of a Hermitian operator supposed to generate the unitary part of the local dynamics, in the context of a specific model of ``universe'' with a simplified initial condition -- but far from equilibrium, with correlations and strong coupling allowed. So motivated, the identified effective Hamiltonian of the two-level system $S$ essentially amounts to an instantaneous frequency given by $\tilde{\omega}^S(t) = -\im{\left(\dot{\psi}(t)/\psi(t)\right)}$, where $\psi(t)$ is the probability amplitude of excitation of $S$. As a consequence, with this definition, the internal energy of an open system is a functional of its quantum state and its first time derivative: $U^S=U(\varrho^S,\dot{\varrho}^S)$. Moreover, although that work addresses only the quantum thermodynamical treatment of $S$, it is easily verified that the application of the same definition of internal energy for the environment, \textit{under the particular hypotheses} of \cite{Valente_etal_2018}, leads to a relation of \textit{consistency} in the sense defined here. 
      
   Motivated by the overview given above, we adopt the following premises to approach the problem of defining the internal energy of an open quantum system ($S$) which composes a bipartite, closed quantum universe. \textit{(i)} The internal energy of $S$ should be given as a functional of all relevant instantaneous parameters of the Universe; \textit{(ii)} the same rule defining internal energy should apply to the environment $E$, yielding a relation of consistency, $U^S+U^E=\avg{H}$; \textit{(iii)} one should be able to write the internal energy of either system as a function of the system's quantum state and its time derivatives.

   We analyze the consequences of those requirements for the simplest non-trivial implementation of a bipartite quantum universe, namely when both $S$ and $E$ are two-level systems (TLS's) in the absence of any external drive. 
\changed{
The choice of a TLS as the environment, the only significant restriction in our framework, is the counterpart of an approximation-free approach to the space of possible consistent internal energy definitions in our closed universe. Although the endeavour of addressing thermodynamical concepts in such a realm might well be put under question, the field of quantum thermodynamics has seen a noteworthy advance of theoretical proposals dismissing the usual restrictions on the size (dimensionality) of the environment, which should then be suited for a two-TLS universe \cite{Weimer_etal_2008, Hossein-Nejad_etal_2015, Alipour_etal_2022}. If a universal approach to quantum thermodynamics is to exist, then it should apply in this case as well; if not, then one should pursue its limits of validity, and taking the simplest model as a point of departure is one possible approach. Indeed, there are well-known examples in recent years of the two-qubit framework, with no room for a larger system and with correlations playing an important role, being adopted for the experimental validation of quantum thermodynamics predictions, namely in NMR systems \cite{Micadei_etal_2019, Pal_etal_2020, Micadei_etal_2021}.  
}
   
   By studying the underlying mathematical structure of \changed{the mentioned} requirements when formalized in \changed{the context of a closed, two-TLS universe evolving from a pure quantum state}, we reduce the problem of the existence of such a definition of internal energy to the (non-)existence of solutions to a certain linear system. We approach the latter problem numerically and derive our main result: if it is the case that a general rule defining 
internal energy
and meeting all the requirements above exists, then the functional dependency must involve, at least, up to the \textit{second-order} time derivative of the local density operator. 

	Moreover, we will show that the definition used in \cite{Valente_etal_2018} remains ``consistent'' for a two-TLS's universe if the interaction and initial conditions are suitably constrained, but not when they are generalized, which is in accordance with our general result. The study of this counter-example might provide some hints to the search for a generally consistent definition of internal energy, and should also motivate further investigation of that particular setting.

	We emphasize that, in our treatment, ``system'' and ``environment'' are modelled on equal footing; there is no equilibrium or weak-coupling assumptions; and the only restriction on the initial state of the Universe is that it be a pure quantum state, so that correlations not only will develop in general, but may be present from the beginning. We also stress that there is no external, classical driving in our description, but our result may be easily generalized to account for driven dynamics, as will be indicated.

	The structure of the text is as follows. We begin by defining the physical context of our analysis and introducing some definitions to implement the ideas sketched in the previous paragraphs (\S \ref{sec:formalism}). Next, we set up the method and main result of this article (\S \ref{sec:gen_res}), moving forward to the study of the above-mentioned particular counter-example (\S \ref{sec:counter_ex}). Then in \S \ref{sec:conclusions} we discuss our conclusions. 

\section{Scope and formalism}\label{sec:formalism}

\subsection{Physical setup and notation}\label{subsec:formalism_setup}

	As anticipated, our \textit{universe} is a closed, autonomous quantum system composed of two interacting TLS's, henceforth labelled $A$ and $B$ to emphasize the absence of any essential distinction between their roles; each may be regarded as the other's ``environment''. We denote by $\mathcal{H}\sj \cong \mathbb{C}^2$ the Hilbert space of system $j$ ($j = A,B$) and thus $\mathcal{H}:=\mathcal{H}\sa\otimes\mathcal{H}\sb$ is that of the universe. Let $\mathcal{L}(\mathcal{X})$ be the space of linear operators on the Hilbert space $\mathcal{X}$ and $\her(\mathcal{X})$ the subset of Hermitian operators. Each subsystem has a bare Hamiltonian $H\sj\in\her(\mathcal{H})$, whose lowest eigenenergy is assumed to be zero (that should not mean any physical restriction); let $\omega\sj>0$ be the remaining eigenvalue. The ground and excited states of $H\sj$ are denoted $\ket{0\sj}$ and $\ket{1\sj}$, respectively, and thus $H\sj = \omega\sj \kb{1\sj}$. By default we take $\mathcal{N}\sj:=\left(\ket{0\sj}, \ket{1\sj}\right)$ as a basis for $\mathcal{H}\sj$, and for $\mathcal{H}$ the corresponding induced basis $\mathcal{N}=\set{\ket{k}}_{k=0,\dots,3}$, labelled by usual binary notation: $\ket{0}:=\ket{00}:=\ket{0\sa}\otimes\ket{0\sb}$, $\ket{1}:=\ket{01}$ and so on. For convenience we also introduce the set of Pauli operators built over $\mathcal{N}\sj$, defined as below:
	
	\begin{subequations}
	\begin{eqnarray}
		\sigma_x\sj := +1\kbb{0\sj}{1\sj} +1\kbb{1\sj}{0\sj}, \\
		\sigma_y\sj := +i\kbb{0\sj}{1\sj} -i\kbb{1\sj}{0\sj}, \\
		\sigma_z\sj := -1\kb{0\sj} +1\kb{1\sj}.
	\end{eqnarray}
	\end{subequations}
		
	Please note that we have ordered $\mathcal{N}\sj$ with the ground state first and defined $\left\lbrace\sigma\sj_k\right\rbrace$ so that $\sigma_z\sj$ assigns the positive eigenvalue to the excited state. %; as a consequence, our Pauli matrices are ``flipped'' with respect to their usual representation. 

	The set $\mathcal{P}\sj = \left\lbrace\mathbb{1}\sj,\left(\sigma_k\sj\right)\right\rbrace$ is a basis for $\mathcal{L}(\mathcal{H}\sj)$
%\cite[\S A-IV]{Cohen_1977} 
and thus the tensor products of $\mathcal{P}\sa$ and $\mathcal{P}\sb$ give a basis for $\mathcal{H}$.
	
	Next we introduce the interaction Hamiltonian, $\hint$, the only restriction on which being that it must effectively represent a (correlation-creating) \textit{interaction}, \ie{}, it should not act (or have any component acting) as an identity on either $\mathcal{H}\sj$. Equivalently, its decomposition in products of $\mathcal{P}\sa$ and $\mathcal{P}\sb$ involves only the $\sigma_k\sj$ and not the identities:
	
	\begin{equation}
		\hint = \sum_{j,k=x,y,z} h_{jk}\sigma\sa _j \otimes \sigma\sb _k,
	\end{equation}
with $h_{jk}\in\mathbb{R}$. We emphasize that the universe is autonomous, \ie{}, $H\sa,H\sb,\hint$ are constant in time.
	
	The universe is supposed to be in a pure state at the starting time $t=0$ and therefore for every $t>0$; its state vector is written 
	
	\begin{equation}
		\ket{\psi(t)}=\sum_{k=0}^3\psi_k(t)\ket{k}, 
	\end{equation}		
and we represent the components $\psi_k$ in polar form,

	\begin{equation}
		\psi_k(t) = R_k(t)e^{i\theta_k(t)}.
	\end{equation}

	The state of each subsystem at any time is described by its density operator, $\varrho\sa = \tr_B\rho, \varrho\sb=\tr_A\rho$, where $\rho:=\kb{\psi}$ is the universe's density operator (time dependencies omitted for clarity). The universe dynamics is of course given by the Schrödinger equation
	
	\begin{equation}\label{eq:formalism:eom}
		\ket{\dot{\psi}} = -iH\ket{\psi}
	\end{equation}		
where

	\begin{equation}
		H = H\sa + H\sb + \hint 
	\end{equation}
stands for the total (universe) Hamiltonian, and the local dynamics are determined by partial trace as indicated above. 

\subsection{Definitions}\label{subsec:formalism:definitions}

	Within the physical setup defined in \S \ref{subsec:formalism_setup}, we will address the problem of defining a functional quantifying the internal energy of the two open systems $A$ and $B$. As anticipated in \S \ref{sec:intro}, we will show that, if such a functional depends only on the local states and their derivatives, and is consistent with the well-known energy of the ``universe'', given by $H$, then it must involve at least up to second-order time derivatives. This will be done in \S\ref{sec:gen_res}. The remainder of this Section is devoted to formalizing the definitions that will be used to state the referred result rigorously.  

	The first one accounts for the following observation. Given the (by now fixed) setup of a universe composed of two interacting TLS's in a pure global state, two instantaneous ``physical situations'' may be considered distinct if and only if at least one of the following takes place: \textit{(i)} the global state vector differs from one to the other, in the sense that they are not connected by a global phase shift; \textit{(ii)} the universe Hamiltonian differs from one to the other. Because the word \textit{state} already has a conventional meaning, we adopt \textit{configuration} to refer to this larger set of variables.
	
    \begin{definition}\label{def:configuration}
    	The \textbf{configuration} (of the universe) at a given time, in which it is at a global state $\ket{\psi}$, having total Hamiltonian $H$, is the pair
    
        \begin{equation}\label{eq:formalism:configuration_definition}
            \mathbf{X} := \left( \ket{\psi}, H\right).
        \end{equation}
    \end{definition}

	This definition is intended to implement the idea of ``everything that may be known of the universe at a given time''. In particular, it makes a distinction between two allowed instantaneous physical situations that may well coincide in state vector but not in Hamiltonian. 
    
    \begin{remark}
    	If two configurations $\mathbf{X}_1$ and $\mathbf{X}_2$ differ only in global phase, \ie{}, if $H_1=H_2$ and $\ket{\psi_1}=e^{i\phi}\ket{\psi_2},\phi\in\mathbb{R}$, then we identify them: $\mathbf{X}_1=\mathbf{X}_2$.
    \end{remark}
    
    \begin{remark}\label{remark:fixed_basis}
	    The basis $\mathcal{N}$ was defined in \S \ref{subsec:formalism_setup} over the eigenstates of the bare Hamiltonians $H\sa,H\sb$. In this sense, it may also change if the configuration changes. However, if two configurations $\mathbf{X}=(\ket{\psi},H)$ and $\mathbf{X}'=(\ket{\psi'},H')$ are such that every component of $\ket{\psi}$ and $H$ in $\mathcal{N}$ coincides with every component of $\ket{\psi'}$ and $H'$ in $\mathcal{N}'$, there is no way to distinguish between them. Abstractly, since our universe is a closed system, all our description is invariant under unitaries in $\mathcal{H}$. Then, for simplicity, we choose to fix $\mathcal{N}$; \ie{}, we map every physical situation into the same abstract basis $\mathcal{N}\subset\mathcal{H}$, to which, therefore, we henceforth refer as an implicitly defined and fixed object. This also motivates the following definition.
    \end{remark}

\iffalse 
    \begin{remark}
    	As pointed out earlier (\S \ref{subsec:formalism_setup}), our physical setup introduces a natural basis for $\mathcal{H}\sa,\mathcal{H}\sb$, and, thence, $\mathcal{H}$; moreover, this basis is constant for a given physical evolution of the universe. Since the latter is a closed system, two hypothetical configurations, not connected by physical evolution, differing \textit{only} in the excited \textit{state} of, \eg{}, $H\sa$, but coinciding in every component of $\ket{\psi}$ and $H$ (each on the respective induced number basis), are physically indistinguishable by any means. This observation justifies identifying a configuration with its representation on the basis $\mathcal{N}$ set up by itself. Synthetically, two configurations that ``differ in their natural bases'', but not in their coordinates on those bases, represent equivalent physical situations and thus are also identified. This will motivate the next definition. 
    \end{remark}
    
    \begin{remark}\label{remark:fixed_basis}
    	From now on, justified by the remark above, we always map every physical situation into the same abstract basis $\mathcal{N}\subset\mathcal{H}$, to which, therefore, we are allowed to refer as an implicitly defined and fixed object.
    \end{remark}
\fi

    \begin{definition}
    	The \textbf{representation of a configuration} $\mathbf{X}=\left(\ket{\psi},H\right)$ is the ordered set of 19 real numbers 
		
		\begin{equation}\label{eq:formalism:config_rep}
		\begin{split}
            X = \left(
                \left\lbrace R_k \right\rbrace_{k=0,\dots,3};
                \left\lbrace \theta_k \right\rbrace_{k=0,\dots,3};
                \right. \\ \left.
                \oma, \omb;
                \left\lbrace h_{jk}\right\rbrace_{j,k=x,y,z}
            \right),
        \end{split} 
		\end{equation}
determined from $\mathbf{X}$ by the prescription of \S\ref{subsec:formalism_setup}. Namely, $R_k\geqslant 0$, $0\leqslant \theta_k <2\pi$ are the usual complex polar coordinates of $\psi_k$, the $k$-th component of $\ket{\psi}$ in the basis $\mathcal{N}$; $\omj$ is the nonzero eigenvalue of $H\sj$; and $\set{h_{jk}}$ are the Pauli matrix components\footnote{We stress that the Pauli bases $\mathcal{P}\sj$ are built over the number bases $\mathcal{N}\sj$.} of $\hint$.
    \end{definition}

    \begin{remark}\label{remark:identification}
    	It should be clear now that two configurations are identical if and only if their representations differ only by a global phase translation, $\set{\theta_k}\mapsto\set{\theta_k+\phi}$.
    \end{remark}
    
    \begin{remark}\label{remark:config_rep_not_injective}
		As per the Remark above, the representation is not properly a function of the configuration, since the same configuration admits a whole family of representations (differing by global phase translation). This is however not a problem for our purposes. On the other hand, in view of Remark \ref{remark:fixed_basis}, the (non-injective) correspondence $X \mapsto \mathbf{X}$ is well defined.
    \end{remark}	

    \begin{definition}
    	The \textbf{configuration space} (of the universe) is the set $\mathbb{S}$ of all possible configurations $\mathbf{X}$, identified pairwise according to Remark \ref{remark:identification}. 
    \end{definition}

	\begin{remark}\label{remark:18dimhype}
		Clearly, the set of all possible \textit{representations} of configurations is the 18-dimensional hypersurface
		
		\begin{widetext}
		\begin{equation}\label{eq:formalism:config_rep_space}
			S = \left\lbrace
            \left(
                \left\lbrace R_k \right\rbrace_{k=0,\dots,3};
                \left\lbrace \theta_k \right\rbrace_{k=0,\dots,3};
                \oma, \omb;
                \left\lbrace h_{jk}\right\rbrace_{j,k=x,y,z}
            \right)	
            %\right. \\ \left. 
            \in\mathbb{R}^{19} : R_k\geqslant 0, \textstyle\sum_kR_k^2=1, 0 \leqslant \theta_k<2\pi, \omj>0
			\right\rbrace.
		\end{equation}
		\end{widetext}
	\end{remark}

	\begin{remark}\label{remark:config_space_17}
		$\mathbb{S}$ is a 17-dimensional manifold, since it is in one-to-one correspondence with $S$ modulo a joint shift of the $\theta_k$.\footnote{More properly, $\mathbb{S}$ is easily seen to be a 17-dimensional manifold by its own construction, since the space of normalized state vectors of $\mathcal{H}\cong \mathbb{C}^4$, identified pairwise when connected by global phase, is the 6-dimensional manifold $\mathbb{C}P^3$ \cite{Ashtekar_1999}, while the Hamiltonian part is trivially isomorphic to $(0,\infty)\times(0,\infty)\times \mathbb{R}^9$.} 
	\end{remark}

	\begin{definition}\label{def:iel}
		An \textbf{internal energy law} (IEL) is a function $\mathcal{E}:\mathbb{S}\to \mathbb{R}^2$. 
	\end{definition}

	\begin{remark}
		Each value in the real ordered pair $\mathcal{E}(\mathbf{X})$ is meant to represent the respective open system's internal energy, in the underlying configuration $\mathbf{X}$ of the universe, according to the law $\mathcal{E}$. We may denote $\mathbf{X} \mapsto \mathcal{E}(\mathbf{X}) = \left(U \sa_\mathcal{E}(\mathbf{X}), U \sb_\mathcal{E}(\mathbf{X})\right)$, or just \eg{} $U\sa_\mathcal{E}$, for shortness, if the configuration is implied. 
	\end{remark}

	\begin{remark}
		As a consequence of our definitions, the result of applying $\mathcal{E}$ to $\mathbf{X}=(\ket{\psi},H)$ cannot depend, for example, on the global phase of $\ket{\psi}$.
	\end{remark}
	
	The definition \ref{def:iel} implements the idea of ``a way to define the internal energy of an open system''. The definition of configuration (\ref{def:configuration}) shows its use here, since it encodes all the variables upon which the internal energy, being an instantaneous notion, could possibly depend; in particular, the configuration determines all the time derivatives of $\ket{\psi}$, since $H$ is a constant. 
	
	In usual quantum thermodynamical constructions, such a functional appears as the usual average of a certain Hermitian operator, either the bare Hamiltonian $H\sj$ \cite{Carrega_etal_2016, Micadei_etal_2019, Ali_Huang_Zhang_2020, Bernardo_2021, Alipour_etal_2022} or some ``renormalized'', effective Hamiltonian $\tilde H\sj$ instantaneously assigned to the system $j$  \cite{Weimer_etal_2008, Hossein-Nejad_etal_2015, Alipour_etal_2016, Ochoa_etal_2016, Dou_etal_2018, Valente_etal_2018, Rivas_2020, Colla_Breuer_2021, Pyharanta_etal_2022}. We stress that this is just a particular way of implementing what we are here defining as an IEL. Conversely, given an IEL, it is trivial to define an observable whose average implements the same rule: $\tilde H \sj (\mathbf{X}) := (\varrho\sj_{11}\omega\sj)^{-1}\cdot U\sj_\mathcal{E}(\mathbf{X}) \cdot H\sj$. Our description is focused on the functional that directly gives the real values $U\sa, U\sb$, which is simpler and suffices to derive our results. 

	The next step is to incorporate the requirement of consistency.
	
	\begin{definition}\label{def:consistency}
		An IEL $\mathcal{E}$ is said to be \textbf{consistent} if, for every $\mathbf{X}\in\mathbb{S}$, the following equation holds:
		
		\begin{equation}\label{eq:formalism:consistency}
			U \sa_\mathcal{E} \left( \mathbf{X} \right)+
			U \sb_\mathcal{E}\left( \mathbf{X} \right)=
			\avg{H}\left( \mathbf{X} \right).
		\end{equation}
	\end{definition}	
	
	Note that $\avg{H}$, as any instantaneous quantity, is indeed a \textit{function} of the configuration: if $\mathbf{X}=(\ket{\psi},H)$, then $
	\avg{H}(\mathbf{X}) := \braket{\psi|H|\psi}	
$.

	It should be emphasized that consistency is a \textit{global} notion, in the sense of the configuration space $\mathbb{S}$. Concretely, then, if an IEL is consistent, it embodies a \textit{universal} way of \textit{consistently} defining internal energy, that is, a rule that respects the quantification of internal energy of the universe regardless of the initial condition and way of coupling between $A$ and $B$.

	At this point we have materialized the first two informal requirements of \S \ref{sec:intro}. The following definitions will account for the third and last one. This step is a little more lengthy, but we may keep in mind that our purpose is just to define the idea of the internal energy ``depending only on the local state and its derivative''.

	\begin{definition}
		The \textbf{1-extended state} of the system $j$ $(= A,B)$ at a given time is the ordered pair of the density operator $\varrho\sj$ and its first time derivative, at that time. We denote 
		
		\begin{equation}
			\boldsymbol{\sigma}\sj_1 := \left( \varrho\sj, \dot{\varrho}\sj \right).
		\end{equation}
	\end{definition}
	
	\begin{remark}\label{remark:sigma_func_x}
		$\boldsymbol{\sigma}_1\sj$ is also a function of $\mathbf{X}$, so we may eventually write $\boldsymbol{\sigma}_1\sj(\mathbf{X})$ (the derivative is given by the equation of motion, \ref{eq:formalism:eom}). Moreover, for shortness, let $\mathbb{\Sigma}_1\sj$ be the set of all physically allowed 1-extended states of $j$; that is to say, all 1-extended states of $j$ that can be obtained from all configurations $\mathbf{X}\in\mathbb{S}$.	
	\end{remark}	

 	We will also need to \textit{represent}	1-extended states by real coordinates. The most convenient choice is to adopt the basis $\mathcal{N}$ (recall Remark \ref{remark:fixed_basis}). 
 	
 	\begin{definition}
 		The \textbf{representation of a 1-extended state} $\boldsymbol{\sigma_1}\sj=(\varrho\sj,\dot{\varrho}\sj)\in\mathbb{\Sigma}_1\sj$ is the set of 6 real numbers

		\begin{equation}
			\sigma\sj_1  = 
			\left( \re{\varrho_{01}\sj}, \im{\varrho_{01}\sj}, \varrho_{11}\sj; 
            \re{\dot{\varrho}_{01}\sj}, \im{\dot{\varrho}_{01}\sj}, \dot{\varrho}_{11}\sj\right)
		\end{equation}
obtained from the matrix representation of $(\varrho\sj,\dot{\varrho}\sj)$ in $\mathcal{N}\sj$.
 	\end{definition}
 	
 	\begin{remark}
 		Of course, every $\boldsymbol{\sigma}_1\sj=(\varrho\sj,\dot{\varrho}\sj)\in\mathbb{\Sigma}_1\sj$ is such that $\varrho\sj,\dot{\varrho}\sj$ are Hermitian and have traces equal to $1$ and $0$, respectively, and therefore the two elements indicated on the representation are sufficient to determine all the remaining matrix elements; the basis $\mathcal{N}\sj$ then suffices to recover the \textit{operators} $\varrho\sj,\dot{\varrho}\sj$.  
 	\end{remark} 
 		
	\begin{remark}\label{remark:ex_state_rep_mappings}
		The representation by itself is a mapping $\boldsymbol{\sigma}_1\sj\mapsto\sigma_1\sj$. This defines implicitly a subset of $\mathbb{R}^6$, namely, the one whose elements may represent 1-extended states of $j$. We denote it $\Sigma_1\sj$. As per the remark above, the inverse mapping $\sigma_1\sj\mapsto\boldsymbol{\sigma}_1\sj$ is well-defined in $\Sigma_1\sj$. Very importantly, here we have a one-to-one correspondence (contrary to the case of $\mathbb{S}$ and $S$, recall Remark \ref{remark:config_rep_not_injective}).
	\end{remark}
	
	\begin{remark}
		It is simple to see that $\Sigma_1\sa=\Sigma_1\sb$; we then indistinctly denote them $\Sigma_1$.
	\end{remark}
	
		It would require additional effort to explicitly characterize the elements of $\Sigma_1$, since the constraint $\varrho \geqslant 0$ is of cumbersome expression in terms of the chosen coordinates and, further, the corresponding constraint on $\dot{\varrho}$ is even more subtle. For our purposes, however, such a characterization is unessential. In particular:

	\begin{remark}\label{remark:ext_state_dimension}
	 The existence of the representation $\boldsymbol{\sigma}_1\sj\mapsto\sigma_1\sj$ implies that the dimensionality of the manifold $\mathbb{\Sigma}_1\sj$ is no greater than 6. 
	\end{remark}

	\begin{definition}\label{def:strong-1l}
		An IEL $\mathcal{E}$ is said to be \textbf{(strongly) 1-local} if the internal energy $U\sj_\mathcal{E}$ can be written as a function of $\varrho\sj$ and its time derivative; that is to say, if there exist two mappings $\hat{\mathcal{E}}\sj: \mathbb{\Sigma}_1\sj \to \mathbb{R}$ such that, for every $\mathbf{X}\in\mathbb{S}$, 
		
		\begin{equation}
			\hat{\mathcal{E}}\sj \left(\boldsymbol{\sigma}_1\sj (\mathbf{X}) \right) = U\sj_\mathcal{E}(\mathbf{X}).
		\end{equation}
	\end{definition}
		
	This definition relates to the idea of a ``rule defining the internal energy of open systems'' (an IEL, Definition \ref{def:iel}) being able to be ``written in terms of local variables'', where \textit{local} means, as in \S\ref{sec:intro}, ``accessible to a measurement apparatus that can only record $\varrho\sj$ as a function of time''. Of course, to account for this idea more properly, one should define an $n$-extended state, with the first $n$ time derivatives of $\varrho\sj$; however, for our purposes, the definitions above suffice: we will investigate the hypothesis that the first derivative is enough. Now rigorously defined, this property is a rather non-trivial one, particularly if one also requires consistency, as we will see shortly (see note after Theorem). The diagram of Fig.  \ref{fig:formalism:strongly_local_diagram} shows the relationships among the three sets and three applications involved in this definition, regarding one of the open systems. 

\ifPacks
	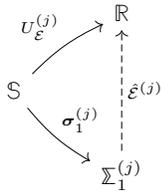
\begin{figure}
	\begin{tikzcd}
		{} & \mathbb{R}  \\
		\mathbb{S} 
			\arrow[ur, "U\sj _\mathcal{E}", bend left=10]
			\arrow[dr, "\boldsymbol{\sigma}\sj_1", bend right=10]		
		& {} \\
		{} & \mathbb{\Sigma}_1\sj
			\arrow[uu, dashed, "\hat{\mathcal{E}}\sj", swap]
	\end{tikzcd}
	\caption{Schematic illustration of the property here named strong 1-locality. Each arrow represents an application between two of the three sets involved: the configuration space, $\mathbb{S}$; the real line, $\mathbb{R}$; and the set of 1-extended states of $j$, $\mathbb{\Sigma}_1\sj$. If the IEL $\mathcal{E}$ is strongly 1-local, then the application $\hat{\mathcal{E}}\sj$ (dashed arrow) exists and is such that $U\sj_\mathcal{E} = \hat{\mathcal{E}}\sj \circ \boldsymbol{\sigma}\sj_1$.}	
	\label{fig:formalism:strongly_local_diagram}
	\end{figure}		
\fi
	
	It cannot be overemphasized that an IEL being strongly 1-local means that it determines  the internal energy of $j$ from sole knowledge of $\varrho\sj, \dot{\varrho}\sj$. It should be clear that the quantities related to the universe Hamiltonian or its components are not conceived as given parameters, but as independent variables (recall Definition \ref{def:configuration}). Then, in our picture, locally assessing the interaction parameters $h_{jk}$, or even the bare energy gap $\omj$, means extracting them from the time derivatives of $\varrho\sj$. (Note that an energy measurement in $j$ is \textit{by hypothesis} not expected to give $\omj$!) That is why a ``0-local'' IEL (a function of $\varrho\sj$ alone) could never have the desired physical meaning, whence we investigate the possibility of 1-locality, the ``minimalist'' hypothesis described in \S \ref{sec:intro}.

	Before we move on to the statement of the main problem to be pursued in this article, we make an additional definition. 
We are not interested in its physical significance, but merely in its logical relationship with the previous one.
	
	\begin{definition}\label{def:weak-1l}
		An IEL $\mathcal{E}$ is said to be \textbf{weakly 1-local} if the internal energy $U\sj_\mathcal{E}$ can be written as a function of $\varrho\sa, \varrho\sb$, and their first time derivatives; that is to say, if there exist two mappings $\hat{\mathcal{E}}\sj: \mathbb{\Sigma}_1\sa \times \mathbb{\Sigma}_1\sb \to \mathbb{R}$ such that, for every $\mathbf{X}\in\mathbb{S}$, 
		
		\begin{equation}\label{eq:formalism:weak-1l}
			\hat{\mathcal{E}}\sj\left(\boldsymbol{\sigma}_1\sa (\mathbf{X}),\boldsymbol{\sigma}_1\sb (\mathbf{X}) \right) = U\sj_\mathcal{E}(\mathbf{X}).
		\end{equation}
	\end{definition}
	
	\begin{remark}
		If an IEL is strongly 1-local, it is weakly 1-local as well.
	\end{remark}
	
\section{General constraint}\label{sec:gen_res}

\subsection{Statement of the problem}

	Put in terms of the definitions just given (\S \ref{subsec:formalism:definitions}), what we have highlighted in \S \ref{sec:intro} about the proposal of \cite{Valente_etal_2018} is that, if worked out, it suggests the possibility of an IEL being consistent and strongly 1-local -- \ie{}, it suggests that one could set up a consistent rule defining internal energy as a function of the open system's quantum state and its first time derivative. Nevertheless, as we will show in \S\ref{sec:counter_ex}, the particular IEL employed there (or, to be precise, its generalized counterpart in the case of a 2-TLS universe) is not consistent in our sense (Definition \ref{def:consistency}), since it yields a consistency relation (Equation \ref{eq:formalism:consistency}) only for a particular, zero-measure class of configurations, corresponding to the specific model interaction and initial conditions analyzed in that work. 
	
	This observation motivates the following question: if not the one employed in \cite{Valente_etal_2018}, is there \textit{some other rule} with \textit{those two simple properties}? As we anticipated in \S \ref{sec:intro}, the answer is negative. Indeed, even a rule that depends on the local states and first derivatives of \textit{both} open systems to define the local energy of \textit{each} of them cannot be truly consistent. That is the main result of the present work, to be achieved in this Section. Formally, we have the following statement.
	
	\begin{prob}\label{prob:strong} 
		\textit{(Strong.)} Find an IEL that is simultaneously \textit{consistent} (Definition \ref{def:consistency}) and \textit{strongly 1-local} (Definition \ref{def:strong-1l}).
	\end{prob}

	Incidentally, our method will show that even a weaker version of this problem is not solvable.
	
	\begin{prob}\label{prob:weak}
		\textit{(Weak.)} Find an IEL that is simultaneously \textit{consistent} (Definition \ref{def:consistency}) and \textit{weakly 1-local} (Definition \ref{def:weak-1l}).
	\end{prob}

\subsection{Method and result}

\subsubsection{Mathematical structure}
	
	To show that Problem \ref{prob:weak} is unsolvable, we begin by making the following observation.
		
	\begin{proposition}\label{prop:if_solvable}
		If Problem \ref{prob:weak} is solvable, then there exists an application $\mathcal{G} : \mathbb{\Sigma}_1\sa \times \mathbb{\Sigma}_1\sb \to \mathbb{R}$ such that, for every $\mathbf{X}\in\mathbb{S}$,
		
		\begin{equation}\label{eq:gen_res:if_solvable}
			\mathcal{G}\left(\boldsymbol{\sigma}_1\sa (\mathbf{X}),\boldsymbol{\sigma}_1\sb (\mathbf{X}) \right) = \avg{H}(\mathbf{X}).
		\end{equation}
	\end{proposition}
	
	\begin{proof}
		This is nearly trivial. Let the IEL $\mathcal{E}$ be consistent and weakly 1-local (hypothesis). Then, for every $\mathbf{X}\in\mathbb{S}$, applying weak 1-locality \eqref{eq:formalism:weak-1l} to the consistency relation \eqref{eq:formalism:consistency} yields, for specific functions $\hat{\mathcal{E}}\sj$, 				
		
		%\begin{widetext}
		\begin{equation}\label{eq:gen_res:if_solvable_proof}
		\begin{split}
\hat{\mathcal{E}}\sa\left(\boldsymbol{\sigma}_1\sa (\mathbf{X}),\boldsymbol{\sigma}_1\sb (\mathbf{X}) \right) + 
\hat{\mathcal{E}}\sb\left(\boldsymbol{\sigma}_1\sa (\mathbf{X}),\boldsymbol{\sigma}_1\sb (\mathbf{X}) \right) 
\\
= \avg{H}(\mathbf{X}),
		\end{split}
		\end{equation}
		%\end{widetext}		
which has the structure of \eqref{eq:gen_res:if_solvable} with $\mathcal{G}:=\hat{\mathcal{E}}\sa +\hat{\mathcal{E}}\sb$.		
		
	\end{proof}
	
	Despite how trivial it may seem, the statement above provides us with a powerful test, insofar as it settles an equality between an \textit{a priori} unknown function of the 1-extended states -- the function $\mathcal{G}$, related to the hypothetical, undetermined IEL $\mathcal{E}$ -- and a well-known function of the configuration -- the average Hamiltonian $\avg{H}$. Intuitively speaking, if we can \textit{change} $\mathbf{X}$ --  ``infinitesimally'', for instance -- in such a way that $\boldsymbol{\sigma}\sa_1(\mathbf{X}),\boldsymbol{\sigma}\sb_1(\mathbf{X})$ \textit{remain unchanged}, but at the same time $\avg{H}(\mathbf{X})$ \textit{changes}, then we can be sure that $\avg{H}$ is not ``a function'' of $(\boldsymbol{\sigma}_1\sa,\boldsymbol{\sigma}_1\sb)$ alone -- that is, a $\mathcal{G}$ as in \eqref{eq:gen_res:if_solvable} cannot exist and, therefore, Problem \ref{prob:weak} must be unsolvable. Essentially, this is the content of our method. The proposition and theorem that follow merely formalize this idea.
	
    \begin{proposition}\label{prop:lin_system_method}
    	Let $U\subset\mathbb{R}^n$ be an open region and $f:U\to\mathbb{R}$, $g:U\to\mathbb{R}^m$ be differentiable maps. Given $x_0\in U$, consider the linear system

        \begin{equation}
            \left\lbrace
            \begin{aligned}
                \Df{x_0} g \cdot \df x &= 0 \\
                \Df{x_0} f \cdot \df x &= \delta f
            \end{aligned}
            \right.
        \end{equation}
in the variable $\df x \in\mathbb{R}^n$, where $0\neq \delta f \in\mathbb{R}$, and $\Df{p}F$ denotes the derivative of $F$ at $p$ (Jacobian matrix). Explicitly, 
    
        \begin{equation}\label{eq:gen_res:lin_system_explicit}
            \left\lbrace
            \begin{aligned}
                \Df{x_0} g_1 \cdot \df x &= 0 \\
                \Df{x_0} g_2 \cdot \df x &= 0 \\
                \vdots \\
                \Df{x_0} g_m \cdot \df x &= 0 \\
                \Df{x_0} f \cdot \df x &= \delta f
            \end{aligned}
            \right.
        \end{equation}
is a system of $m+1$ equations in the $n$ variables $\df x = \left(\df x_1, \hdots, \df x_n\right)\in\mathbb{R}^n$. Under these conditions, if there exists an $\hat{f}:g(U)\subset\mathbb{R}^m\to\mathbb{R}$ such that $\hat{f}(g(x)) = f(x)$ for every $x\in U$, then, for any $x_0\in U$, the system above is inconsistent.
    \end{proposition}
    
	\begin{remarkprop}
	This proposition is of particular interest when $1<m<n$, so that $g$ is typically not invertible and the existence of an $\hat{f}$ as above is nontrivial.
	\end{remarkprop}
    
    \begin{remarkprop}\label{remarkprop:test_df_1}
    Provided $\delta f \neq 0$, its value is clearly irrelevant for the question above, so we may set $\delta f = 1$ for simplicity.
    \end{remarkprop}
    
    \begin{proof}
		See Appendix \ref{sec:app:proof_method}.
		
    \end{proof}
    
    Finally, we use Propositions \ref{prop:if_solvable} and \ref{prop:lin_system_method} together to obtain a verifiable signature of the nonexistence of an IEL satisfying Problem \ref{prob:weak}. 
	
	\begin{theorem*}\label{prop:lin_syst_applied}
		If Problem \ref{prob:weak} is solvable, then the linear system

        \begin{equation}\label{eq:gen_result:system}
            \left\lbrace
            \begin{aligned}
                \Df{X_0} \sigma\sa_1 \cdot \df X &= 0 \\
                \Df{X_0} \sigma\sb_1 \cdot \df X &= 0\\
                \Df{X_0} \left( \sum_{k=0}^3 R_k^2 \right) \cdot \df X  &= 0\\
                \Df{X_0} \avg{H} \cdot \df X &= \delta E  \neq 0\\
            \end{aligned}
            \right.
        \end{equation}
in the 19 variables 

        \begin{equation}
        \begin{split}
            \df X = \left(
                \left\lbrace \df R_k \right\rbrace_{k=0,\dots,3};
                \left\lbrace \df \theta_k \right\rbrace_{k=0,\dots,3};
                \right. \\ \left.
                \df \oma, \df \omb; 
                \left\lbrace \df h_{jk}\right\rbrace_{j,k=x,y,z}
            \right) \in \mathbb{R}^{19}
		\end{split}        
        \end{equation}
is inconsistent at every configuration representation $X_0\in \int(S)$, where $\int(S)$ is the interior of $S$ (Equation \ref{eq:formalism:config_rep_space}, except for points with some $R_k=0$ or $\theta_k=0$). 
	\end{theorem*}

	\begin{proof}
	
		The statement is simply what results from applying Propositions \ref{prop:if_solvable} and \ref{prop:lin_system_method} to our context of interest. Although it is quite intuitive, a formal proof is worthwhile. We defer it to Appendix \ref{sec:app:proof_applied}.
		
	\end{proof}

	\begin{remarkthm*}\label{remarkprop:probably_not_solvable}
		\eqref{eq:gen_result:system} is a system of 14 equations in 19 variables; ordinarily such a system is expected to have infinite solutions. At the present level of abstraction, therefore, it has become clear that Problem \ref{prob:weak} is \textit{most likely impossible} to solve; in other words, we can already sense that \textit{being consistent} is an extremely \textit{strong} requirement for a \textit{weakly 1-local} IEL, and vice-versa. Of course, this is reminiscent of the fact that the manifold $\mathbb{\Sigma}_1\sa\times\mathbb{\Sigma}_1\sb$ has (no more than) 12 dimensions (Remark \ref{remark:ext_state_dimension}), whereas the configuration space $\mathbf{X}$ has 17 (Remark \ref{remark:config_space_17}); one should not expect to be able to ``recover'' the configuration given only the pair of 1-extended states and thus, in principle, nor to generally recover a particular functional of the former. 
		
		Of course, however, some ``hidden similarity'' between the structures of the functions $\sigma_1\sj(X)$ and $\avg{H}(X)$ could result in the opposite, counter-intuitive case happening, \ie{}, the last row of the coefficient matrix of \eqref{eq:gen_result:system} could turn out to be a linear combination of the preceding ones, resulting in an inconsistent system -- and that is why we go further towards ruling out such possibility. Since the number of variables is too large to be  worked out analytically or even by symbolic computation, we resort to a numerical approach.
	\end{remarkthm*}

\subsubsection{Numerical method}
	
	With the result above, we find ourselves in a position to numerically test whether a solution to Problem \ref{prob:weak} may exist. Strictly speaking, that could be achieved by merely choosing \textit{one} point $X_0\in\int S$, even arbitrarily; computing the matrix elements of the linear system \eqref{eq:gen_result:system}; and showing that it has at least \textit{one} solution. That being done, we would be sure that Problem \ref{prob:weak} is unsolvable, \ie{}, that an IEL cannot be consistent and weakly 1-local. More precisely, we would be finding that the referred properties could not be simultaneously satisfied by any single IEL in the entire ``bulk'' of $\mathbb{S}$, which by itself would indeed rule out those two properties as defined in \S \ref{subsec:formalism:definitions}, insofar as they refer to ``global'' relationships (recall Definitions \ref{def:consistency} and \ref{def:weak-1l}). Nevertheless, it is quite simple to go further: if we can generate an appreciable number of well-distributed points $X_0$ across $\int S$ and show that the system \eqref{eq:gen_result:system} is solvable in \textit{all} of them, we will be finding strong evidence that even an IEL satisfying those two properties ``locally'', in \textit{``bulk'' regions} (\ie{}, open, thence nonzero-measure, subsets) of $\mathbb{S}$, cannot exist. 
	
	We then proceed as sketched above. Since we are not to systematically discuss the choice of the probability distribution in $\int S$, nor which number of generated points should be considered ``large'', we shall keep the latest, strongest conclusion anticipated in the previous paragraph in the level of a ``strong evidence'', whereas the first one is clearly more rigorous. Our numerical routine, implemented via a Python script, consists essentially of the following steps:
	
	\begin{enumerate}[label=\arabic*)]
		\item Randomly generate a point $X_0 \in \int S$ by drawing on a uniform distribution on the respective range of each of the 19 real coordinates (Equation \ref{eq:formalism:config_rep_space}); 
		\item Build the matrix \texttt{M} of the system \eqref{eq:gen_result:system} by numerically differentiating $\sigma_1\sj, \sum_kR_k^2$, and $\avg{H}$, as functions of $X$, and evaluating the results at $X_0$; 
		\item Call a least-squares routine to numerically find a ``candidate'' particular solution \texttt{dX} to the system \texttt{M.dX = b}, where \texttt{b} is a 19-list filled with zeros except for the last entry, which is set to one (recall Prop. \ref{prop:lin_system_method}, Note \ref{remarkprop:test_df_1}); 
		\item Evaluate \texttt{M.dX - b} and check whether its Euclidian norm is zero to within a certain absolute threshold;
		\item Repeat until the desired amount of points is achieved.
	\end{enumerate}

	Two observations regarding the procedure above should be made:
	
	\begin{enumerate}[label=\alph*.]
		\item In principle, the coordinates $\omega\sj$, $h_{jk}$ should be drawn from unbounded intervals. Although this is by no means impossible, we choose to keep the simplicity of using uniform distributions on, respectively, the intervals $(0,1)$ and $(-1,1)$. This should not represent any loss of generality, since any choice of those 11 energy values within their original ranges should be physically indistinguishable, at least insofar as only the properties here investigated are concerned, from another one, obtained from the first by ``normalizing'' by its largest absolute value. In short, only the ratios of energies should be relevant and therefore we can sample only in the bounded intervals just defined. 
		\item The routine involves two computational parameters, in principle arbitrary: \textit{(i)} the accuracy $h$ of numerical differentiation, and \textit{(ii)} the threshold under which the remainder \texttt{M.dX - b} is considered zero.  Because we employ two-point central-difference numerical derivative and computations are performed with double floating point precision, the estimated optimal choice for $h$ is about $10^{-6}$, which we adopt.\footnote{See \eg{} \cite{RH_Landau_etal_2015}, \S 5.5.} With this choice, the precision in the matrix elements of \texttt{M} should be expected to be no better than about $10^{-13}$,\footnote{\textit{Id.}} so that we choose a threshold $10^{-12}$ for the norm of the remainder.
	\end{enumerate}		
		
\subsubsection{Result}

	Under the settings above, with a sample of $5\times 10^{3}$ points $X_0$, we found that the system always has a solution. Even shrinking the threshold to $10^{-13}$, close to machine precision, the result remains. The Jupyter notebook is available upon request.
	
\subsubsection{Summary}\label{subsec:result_summary}

	In this Section, we have proposed the problem of the existence of an internal energy law (IEL) satisfying simultaneously the (global) requirements of consistency and weak 1-locality, all defined in \S \ref{subsec:formalism:definitions}. We then showed that the existence of one such IEL would imply that a certain linear system (Equation \ref{eq:gen_result:system}), with 5 variables more than equations, be \textit{unsolvable} at every configuration representation $X_0\in \int S$. We then tested such hypothesis numerically, showing that the system consistently turns out to be solvable, to the best accuracy attainable with standard machine precision, for $5\times 10^{3}$ randomly sampled representations. This result further suggests that even an IEL satisfying locally the requirements of consistency and weak 1-locality in any open (and, thus, nonzero-measure) subset of $\mathbb{S}$ must not exist. 
	
	As already pointed at the beginning of this Section, though, we do know of an instance of an IEL satisfying the referred properties ``locally'', under a very particular physical setting, which turns out to represent a zero-measure subset of the configuration space. The next Section is devoted to this counter-example.

\subsubsection{Addendum: driven dynamics}

	All our theoretical framework so far has been built upon the hypothesis of an autonomous universe, \ie{}, the Hamiltonian $H$ was assumed time-independent, which is justified by the point of view adopted since \S \ref{sec:intro}: we are interested in the energy exchange \textit{between quantum systems}; or, equivalently, we recognize that, at the most elementary level, all relevant physical entities should be described quantum-mechanically. Nevertheless, it is a simple task to generalize the formalism to account for classically driven systems, \ie{}, time-dependent Hamiltonians. To avoid complications, though, one must suppose that the individual Hamiltonians $H\sj$ vary only in their energy gap $\omega\sj$, \ie{}, their eigenvectors should still be assumed constant. With this sole restriction, our result, that a consistent IEL cannot be weakly 1-local, is easily seen to remain valid. Although the definition of configuration as made here is of more limited use in the time-dependent case, insofar as the time derivatives of $H$ would need to be included if higher-order derivatives of the $\varrho\sj$ were concerned, such inclusion is clearly immaterial if one is to derive a result on 1-local IELs, and then all steps through the Theorem can be taken identically.

\section{A counter-example}\label{sec:counter_ex}

	The definition of internal energy found in \cite{Valente_etal_2018} is strongly 1-local in our sense (Def. \ref{def:strong-1l}). The study concerns a particular physical model, of a TLS coupled to an electromagnetic field via exchange interaction, with a particular initial condition, of a pure universe state restricted to the zero- and one-excitation subspaces; under these special conditions, a relation that we would name consistency (Def. \ref{def:consistency}) can be derived.

	In the present work, the universe comprises \textit{two} TLS's, each playing the role of an ``environment'' to the other; nevertheless, it is straightforward to ``mimic'' the interaction Hamiltonian, initial state, and effective Hamiltonian of \cite{Valente_etal_2018} in this different setup. This is done in the present Section. We first introduce the energy renormalization law (IEL) analogous to that of \cite{Valente_etal_2018}, motivated by the limiting case of uncoupled systems (\S \ref{subsection:counter_ex:uncoupled}); then we introduce the interaction Hamiltonian (\S \ref{subsec:counter_ex:ham}) and initial condition (\S\ref{subsec:counter_ex:dynamics}), showing that a ``local'' relation of consistency is indeed reproduced here, but that it loses validity if the interaction \textit{or} initial condition is modified (\S \ref{subsec:counter_ex:consistency}). The goal is to illustrate the abstract concepts and definitions involved in our main result, as well as to indicate possible guidelines for further investigation. 

\subsection{Uncoupled systems}\label{subsection:counter_ex:uncoupled}

	We first consider the limiting case of two non-interacting TLS's, $\hint=0$, for which there is no doubt as to what operator represents the physical energy of each system. In this case the equations of motion (\S \ref{subsec:formalism_setup}) give

        \begin{equation}\label{eq:counter_ex:dt_rho_a_uncoup}
            \dot{\varrho}\sa =
            \begin{pmatrix}
                0 & i\oma (\psi_0\psi_2^*+\psi_1\psi_3^*)\\
                * & 0 
            \end{pmatrix},
        \end{equation}
where the term $_{10}$, given by hermiticity, was omitted for clarity. (The analogous equation for $B$ is easily obtained by the symmetry $A \leftrightarrow B, \psi_1\leftrightarrow \psi_2$.) In turn, regardless of $\hint$,

        \begin{equation}\label{eq:counter_ex:rho_a}
            \varrho\sa =
                \begin{pmatrix}
                    \absq{\psi_{0}}+\absq{\psi_{1}} &
                    \psi_0\psi_2^*+\psi_1\psi_3^* \\
                 * & 
                    \absq{\psi_{2}}+\absq{\psi_{3}}
                \end{pmatrix}   .
        \end{equation}
        
	We readily see that each TLS ($j$) displays a very simple relationship between the matrix elements of $\varrho\sj$ and $\dot{\varrho}\sj$, on the one hand, and the bare Hamiltonian $H\sj$, on the other. In fact, Equations \eqref{eq:counter_ex:dt_rho_a_uncoup} and \eqref{eq:counter_ex:rho_a} yield $\dot{\varrho}\sa_{01} = i\oma \varrho_{01}\sa$, which is immediately obtained for $B$ as well, so that
	
        \begin{equation}
            \omj = \im{\left( 
            \frac{\dot{\varrho}\sj_{01}}{\varrho\sj_{01}}
            \right)}.
        \end{equation}
        
	Therefore, the internal energy law given by
	
		\begin{equation}
		\begin{aligned}\label{eq:counter_ex:rc_ham}
            \sigma\sj \mapsto 
            U_{RC}\sj &:= \tr \left(\varrho\sj \tilde{H}_{RC}\sj \right), \\
            \tilde{H}_{RC}\sj &:=
            \tilde{\omega}(\sigma\sj) \kb{1\sj},  \\		
			\tilde{\omega}(\sigma\sj) &:= 
			\im{\left( 
				\frac{\dot{\varrho}\sj_{01}}{\varrho\sj_{01}}
			\right)},
		\end{aligned}
		\end{equation}
is such that, for \textit{non-interacting} systems, 
    
        \begin{equation}\label{eq:counter_ex:rc_uncoupled}
            \tilde{H}_{RC}\sj = H\sj.
        \end{equation}

	We adopt the label \textit{rotating coherence} (RC) to refer to this particular IEL and the corresponding effective Hamiltonians, in reference to the fact that they ``detect'' the transition frequency $\omega\sj$ as the rotation frequency of the phasor $\varrho\sj_{01}$ -- the coherence of the density matrix $\varrho\sj$ in the natural basis $\mathcal{N}\sj$ (\S \ref{subsec:formalism_setup}) -- in the complex plane. 

	The RC internal energy law is the first (and only) concrete instance of an IEL to be considered in this article. We have just seen that it recovers the bare Hamiltonians in the limiting case of no interaction, which is desired. Moreover, it clearly is \textit{strongly 1-local}, by definition (Def. \ref{def:strong-1l}). And, from Equation \eqref{eq:counter_ex:rc_uncoupled}, it follows readily that it satisfies a consistency relation: if $U_{RC}:=U\sa_{RC}+U\sb_{RC}$, then in the no-interaction limit
	
		\begin{equation}\label{eq:counter_ex:consistency_uncoupled}
		\begin{aligned}
			U_{RC} &= 
				\tr \left(\varrho\sa \tilde{H}_{RC}\sa \right)+
				\tr \left(\varrho\sb \tilde{H}_{RC}\sb \right)\\
				&= \tr \left(\varrho\sa H\sa \right)+
				\tr \left(\varrho\sb H\sb \right)\\
				&=\avg{H\sa}+\avg{H\sb}\\
				&=\avg{H}.	
					\hspace{60pt}(\hint=0)				
		\end{aligned}
		\end{equation}

	According to our result (\S \ref{subsec:result_summary}), such a relation cannot hold in general, \ie{}, for arbitrary interaction and initial condition. However, as motivated by \cite{Valente_etal_2018}, where the proposed IEL is essentially analogous, we will check that, for suitable interaction and initial condition, the RC rule still exhibits a non-trivial consistency relation. We can also see the breakdown of consistency under the slightest change in the interaction model or initial condition. This is the matter of the next topic.

\subsection{Excitation-conserving interaction}\label{subsec:counter_ex:excit_cons}

\subsubsection{Hamiltonian}\label{subsec:counter_ex:ham}

    We will address the question of consistency of the RC internal energy under the following interaction Hamiltonian:

        \begin{equation}\label{eq:counter_ex:num_ham}
            H^\text{num.}_{\Lambda,\Delta}
            := \Lambda\sigma_+\sa\sigma_-\sb
            + \Lambda^*\sigma_-\sa\sigma_+\sb
            +\Delta\sigma_z\sa\sigma_z\sb,
        \end{equation}
parameterized by $\Lambda\in\mathbb{C},\Delta\in\mathbb{R}$, where $\sigma_+\sj, \sigma_-\sj=\sigma_+^{(j)\dagger}$ are the usual pseudospin operators, $\sigma_+\sj = \kbb{1\sj}{0\sj}$. By writing the explicit matrix forms, it is not difficult to verify that $H^\text{num.}_{\Lambda,\Delta}$ is the most general interaction Hamiltonian that commutes with the \textit{total number of excitations} operator, $N:=N\sa+N\sb =  \kb{1\sa}\otimes \mathbb{1}\sb +\mathbb{1}\sa\otimes\kb{1\sb} = \left.\text{diag}\left(0,1,1,2\right)\right|_{\mathcal{N}}$. It follows, in particular, that the observable $N$ is conserved under the evolution generated by $H := H\sa+H\sb+H^\text{num.}_{\Lambda,\Delta}$. 

	We shall also consider the particular cases obtained by making $\Delta = 0$ or $\Lambda=0$ in \eqref{eq:counter_ex:num_ham} -- respectively, the \textit{exchange} and \textit{dephasing} interactions,
	
		\begin{align}
            H^\text{ex.}_{\Lambda}
            &:= \Lambda\sigma_+\sa\sigma_-\sb
            + \Lambda^*\sigma_-\sa\sigma_+\sb, \label{eq:counter_ex:ex_ham}
            \\
            H^\text{dep.}_{\Delta}
            &:=
            \Delta\sigma_z\sa\sigma_z\sb. \label{eq:counter_ex:dep_ham}
		\end{align}
		
	The former is analogous to the Jaynes-Cummings model of interaction between a TLS and a harmonic mode, a version of which was adopted in \cite{Valente_etal_2018}; the latter can be thought of as an Ising-like dipole coupling between two spins (in the same direction of the local fields), and we include it in our interaction model for reasons of contrast, to become clear later.

\subsubsection{Dynamics in the absence of double excitation}\label{subsec:counter_ex:dynamics}
	
	To derive our analytical results, we make a restriction analogous to that found in \cite{Valente_etal_2018}: the initial state of the Universe is taken to be $\ket{\psi(0)} = \psi_0(0)\ket{00} + \psi_1(0)\ket{01} + \psi_2(0)\ket{10}$. (We will also show results of numerical simulations for a more general initial state, \S \ref{subsec:counter_ex:consistency}.) The already mentioned symmetry of $H^\text{num.}_{\Lambda,\Delta}$ then implies that, for any $t\geqslant 0$, 

        \begin{equation}\label{eq:counter_ex:dynamics:universe_state}
			\ket{\psi(t)} = \psi_0(t)\ket{00} + \psi_B(t)\ket{01} + \psi_A(t)\ket{10},
        \end{equation}
where we suggestively adopted the notations $\psi_A:=\psi_2,\psi_B:=\psi_1$, since these quantities will give the respective excited populations.
        
	We intend solely to calculate the quantities $U\sa_{RC},U\sb_{RC}$, and $\avg{H}$. The crucial equation is \eqref{eq:counter_ex:rc_ham}; we then need the matrix elements of $\varrho\sj,\dot{\varrho}\sj$, in terms of the $\psi_k$ and the parameters entering $H$.\footnote{This corresponds essentially to writing the 1-extended states as explicit functions of the configuration (\S \ref{subsec:formalism:definitions}).} The former was already given in Equation \eqref{eq:counter_ex:rho_a} -- it is obviously unaffected by the interaction; we just particularize for $\psi_3=0$ (time dependencies suppressed):
	
		\begin{subequations}\label{eq:counter_ex:dynamics:local_states}
        \begin{align}
            \varrho\sa &=
            \begin{pmatrix}
                1-\absq{\psi_A} & \psi_0\psi_A^*\\
                * & \absq{\psi_A}
            \end{pmatrix},
\\
            \varrho\sb &=
            \begin{pmatrix}
                1-\absq{\psi_B} & \psi_0\psi_B^*\\
                * & \absq{\psi_B}
            \end{pmatrix}.
        \end{align}
		\end{subequations}
	
	The calculations leading to the $\dot{\varrho}\sj$ are considerably simplified by the fact that $\psi_3=0$. We defer them to Appendix \ref{sec:app:counter_ex_explicit}; the results are
	
		\begin{subequations}\label{eq:counter_ex:dynamics:local_derivatives}
        \begin{align}
            \dot{\varrho}\sa &= \trp{B}\dot{\rho} = 
            \begin{pmatrix}
                *&
                 i\psi_0\left[
                    \Lambda^* \psi_B^*+
                    (\oma-2\Delta)\psi_A^*
                \right]\\
                *&
                +2\im{\left(\Lambda\psi_A^*\psi_B\right)}
            \end{pmatrix},
\\
            \dot{\varrho}\sb &= \trp{A}\dot{\rho} = 
            \begin{pmatrix}
                *&
                 i\psi_0\left[
                    \Lambda \psi_A^*+
                    (\omb-2\Delta)\psi_B^*
                \right]\\
                *&
                -2\im{\left(\Lambda\psi_A^*\psi_B\right)}
            \end{pmatrix}.
        \end{align}
		\end{subequations}

	The suppressed terms are implicit since $\dot{\varrho}=\dot{\varrho}\hc, \tr \dot{\varrho}=0$. As expected, each equation can be derived from the other by the symmetry $A\leftrightarrow B, \Lambda\leftrightarrow \Lambda^*$.
    
\subsubsection{(In)consistency relations}\label{subsec:counter_ex:consistency}
	
	We now apply Equations \eqref{eq:counter_ex:rc_ham} to this particular case. With \eqref{eq:counter_ex:dynamics:local_states} and \eqref{eq:counter_ex:dynamics:local_derivatives},

        \begin{align}
            \tilde{\omega}(\sigma\sa) &= \im{\left( 
                \frac{
                    i\psi_0\left[
                    \Lambda^* \psi_B^*+
                    (\oma-2\Delta)\psi_A^*
                \right]      
                }{
                    \psi_0\psi_A^*
                }
                \right)}\nonumber \\
                &= \oma - 2\Delta + \re{\left(
                    \Lambda\frac{\psi_B}{\psi_A}
                \right)},
        \end{align}
so that $U\sa_{RC} := \avg{\tilde{H}_{RC}\sa} = \tr \left(\varrho\sa \ef \omega(\sigma\sa) \kb{1\sa} \right) = \absq{\psi_A}\tilde{\omega}(\sigma\sa)$ gives

        \begin{equation}
            U\sa_{RC} = \absq{\psi_A}(\oma - 2\Delta) + \re{\left(
                    \Lambda \psi_B \psi_A^*
                \right)};
        \end{equation}
in turn, ($A\leftrightarrow B, \Lambda\leftrightarrow \Lambda^*$)
    
        \begin{equation}
            U\sb_{RC} =\absq{\psi_B}(\omb - 2\Delta)
            + \re{\left(
                    \Lambda \psi_B \psi_A^*
                \right)}.
        \end{equation}
        
    Therefore, the total internal energy, according to the RC law, results
    
        \begin{equation}\label{eq:counter_ex:consistency:u_rc}
		\begin{aligned}
            U_{RC} &= \absq{\psi_A}\oma 
            + \absq{\psi_B}\omb \\
            &+ 2\re{\left(
                    \Lambda  \psi_A^* \psi_B
                \right)}
            - 2(1-\absq{\psi_0})\Delta,
		\end{aligned}        
        \end{equation}
where we used normalization, $\absq{\psi_0}+\absq{\psi_A}+\absq{\psi_B}=1$. In turn, the average of $H=H\sa+H\sb+H^\text{num.}_{\Lambda,\Delta}$ gives (\textit{cf.} Appendix \ref{sec:app:counter_ex_explicit})

        \begin{equation}\label{eq:counter_ex:consistency:mean_h}
        \begin{aligned}
            \avg{H} &= \absq{\psi_A}\oma 
            + \absq{\psi_B}\omb \\
            &+ 2\re{\left(
                    \Lambda  \psi_A^* \psi_B
                \right)}
            - 2(1-\absq{\psi_0})\Delta
            + \Delta.
        \end{aligned}
        \end{equation}
        
	The two preceding equations then give, for the number-conserving interaction Hamiltonian (Equation \ref{eq:counter_ex:num_ham}), 
        
        \begin{equation}\label{eq:counter_ex:consistency:inconsistency}
            U_{RC} = \avg{H} - \Delta,
        \end{equation}        
where we should have expected $U_{RC} = \avg{H}$ (consistency). Equation \eqref{eq:counter_ex:consistency:inconsistency}, besides generalizing \eqref{eq:counter_ex:consistency_uncoupled}, shows, at once, that:

	\begin{enumerate}[label=\textit{(\roman*)}]
		\item Indeed the rotating-coherence IEL \eqref{eq:counter_ex:rc_ham} gives a ``local consistency'' relation for the exchange interaction (\ref{eq:counter_ex:ex_ham}, or \ref{eq:counter_ex:num_ham} with $\Delta=0$) with an initial condition bound to the $[N = 0] \oplus [N = 1]$ subspace -- which corresponds exactly to ``mimicking'', in our two-TLS universe, the conditions of \cite{Valente_etal_2018} for which such a relation may be found; 
		\item nevertheless, the same relation ceases to hold even if we maintain the special characteristic of no-double-excitation dynamics but switch to the more general Hamiltonian \eqref{eq:counter_ex:num_ham}. 
	\end{enumerate}
		
	With regards to the last point above, we observe that it seems unlikely that one could correct the IEL \eqref{eq:counter_ex:rc_ham} so as to make an additional $\Delta$ come up in the average of $\ef H$, since for an (even weakly) 1-local IEL this parameter should be ``extracted'' from the pair \eqref{eq:counter_ex:dynamics:local_derivatives}, where  it appears ``tied'' to other terms, particularly the $\omega\sj$. 
	
	As a matter of illustration, we show in Fig. \ref{fig:counter_ex:simulations} results of numerical simulation of the discussed quantities for the interaction \eqref{eq:counter_ex:num_ham}, switching from consistency to non-consistency; for completeness, besides reproducing the analytical result of Equation \eqref{eq:counter_ex:consistency:inconsistency}, we also simulated an initial condition with $\psi_3(0)\neq 0$, for which no calculation was performed. 
	
	\begin{figure}                                                                                                                                                                                                                                                                                                                                                                                    	\includegraphics[width=\linewidth]{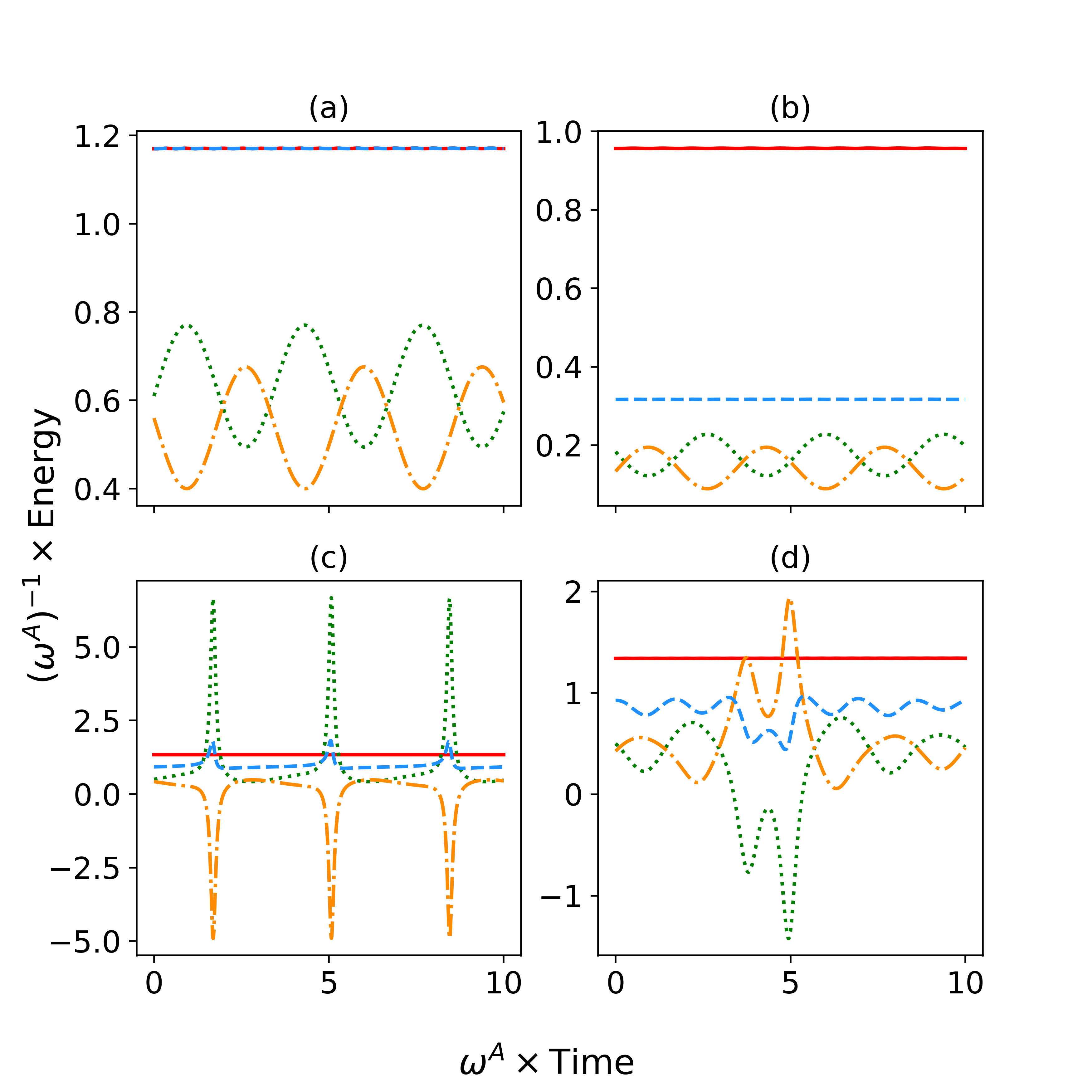}	
	\caption{
		Numerical simulation of the RC law, Eq. \eqref{eq:counter_ex:rc_ham}, showing the quantities $U\sa_{RC}$ (dotted green), $U\sb_{RC}$ (dashed/dotted orange), $U_{RC}$ (dashed blue), and $\avg{H}$ (red). The ``bare'' frequencies are $\oma = 1$ and $\omb = 0.85$; the interaction Hamiltonian is \eqref{eq:counter_ex:num_ham}, with $\Lambda = 0.83+0.41i$ and switching values of $\Delta$: (a) and (c): $\Delta = 0$; (b) and (d): $\Delta = 0.64$. The initial state is $\ket{\psi(0)}\propto \ket{00}+\ket{01}+\ket{10}+\alpha\ket{11}$, with $\alpha=0$ in (a) and (b) and $\alpha=1$ in (c) and (d). The upper half then illustrates Eq. \eqref{eq:counter_ex:consistency:inconsistency}: for $\psi_3(0)=0$, RC yields ``consistency'' if $\Delta = 0$, but an offset $-\Delta$ otherwise. The lower half shows that for more general initial conditions the quantity $U_{RC}$ actually becomes time-dependent, even for autonomous dynamics, $\df H/\df t = 0$. 
	} 
	\label{fig:counter_ex:simulations}
	\end{figure}

\subsubsection{Discussion}

	In this Section, we introduced a strongly 1-local IEL that seems natural to define in the limiting case of uncoupled TLS's. We also showed that it appears to be ``consistent'' if we analyze only the particular case of the exchange interaction \eqref{eq:counter_ex:ex_ham} in the absence of double excitation \eqref{eq:counter_ex:dynamics:universe_state} -- a setting to which we henceforth refer as \textit{the control case} -- but that this ``consistency'' breaks down if either restriction is given up. 
	
	First, it should be stressed that the emergence of a consistency relation for the RC law in the control case does not at any rate contradict the general constraint derived in \S \ref{sec:gen_res}. The referred IEL was actually shown to be \textit{not consistent}, insofar as our notion of consistency (Def. \ref{def:consistency}) is a \textit{global} one: it supposes certain equality to hold for \textit{every configuration} $\mathbf{X}\in\mathbb{S}$, which in practice means \textit{for any interaction Hamiltonian and (pure) initial universe state}. 
	
	We can go further and observe that even the stronger result claimed in \S \ref{subsec:result_summary}, that a weakly 1-local IEL cannot satisfy a consistency relation in \textit{any open region} of $\mathbb{S}$, remains valid: \textit{the control case corresponds to a zero-measure subset of} $\mathbb{S}$ -- the one defined by $R_k=0$ and a particular set of constraints on the $h_{jk}$ --, thus not an open subset, and therefore it is not contradictory that ``local consistency'' can arise for a 1-local IEL in the control case. 

	Notwithstanding the fragility of the ``consistency'' of the RC Hamiltonian, it is noteworthy that such a property is able to be found in a strongly 1-local IEL for a nontrivial, physically meaningful setup, namely the control case. In fact, we have not found a different setting (interaction and initial state) to exhibit such ``singular'' phenomenon, which naturally leads to wondering whether there exists some particular physical property in the control case that makes it ``special'' in this sense. If the answer were positive, one could speculate this particular setup -- or its counterpart in a universe comprised of a TLS in a boson bath, studied in \cite{Valente_etal_2018} -- to be particularly ``appealing'' for addressing quantum thermodynamical questions. 
	
	In trying to figure out what special physical feature of the control case separates it from the general case in which consistency of the rotating-coherence law \eqref{eq:counter_ex:rc_ham} breaks down, we have not as yet been able to find a conclusive answer. 
The following hypothetical explanations were considered and should be discarded:

	\begin{enumerate}[label=\textit{(\alph*)}]
		\item The obvious symmetry $\left[H^\text{ex.}_{\Lambda}, N \right] = 0$, conjugated with the hypothesis $\psi_3(0)=0$, which imply the particular form of the state vector \eqref{eq:counter_ex:dynamics:universe_state}: this feature does not distinguish the control case from its ``extended'' version with $\hint = H^\text{num.}_{\Lambda,\Delta}, \Delta \neq 0$, for which the RC law already loses its ``consistency'', Equation \eqref{eq:counter_ex:consistency:inconsistency};
		\item Some ``hidden'' symmetry/conserved quantity, peculiar to $H^\text{ex.}_{\Lambda}$: actually, it is possible to show, by explicit construction, that if $\mathcal{O} \in \her (\mathcal{H})$ commutes with $H^\text{ex.}_{\Lambda}$, then $\mathcal{O}$ automatically commutes with $H^\text{dep.}_{\Delta}$ as well, and thus again we did not ``isolate'' the control case;
		\item The low dimensionality of the accessible region of $\mathbb{S}$: the constraints $\psi_3=0,\hint=H^\text{ex.}_{\Lambda}$ define an 8-dimensional submanifold of $\mathbb{S}$, which perhaps could be rendering the inversion $\mathbf{X}\mapsto (\boldsymbol{\sigma}_1\sa,\boldsymbol{\sigma}_1\sb)$ possible in this particular case; however, the same dimensionality reduction is achieved with $\hint = H^\text{num.}_{\Lambda,\Delta}$ if for example we impose $\im \Lambda =0$, and thus, again, the control case is not isolated.
	\end{enumerate}
	
	We thus leave the suggestion for future investigation to seek a more physically appealing characterization of what we have defined as the control case, by means, \eg{}, of some informational property, and its connection with the possibility of establishing a strongly 1-local IEL in that particular setting.

\section{Concluding remarks}\label{sec:conclusions}

\iffalse 
\begin{quote}\textit{
[...] \textbf{macroscopic} systems have definite and precise energies, subject to a definite conservation principle.
}\end{quote}
\begin{flushright}
H. B. Callen (\cite{Callen_1985}, p. 12; highlight by the authors)
\end{flushright}
\fi 

	Appreciable effort has been devoted in recent years to addressing a thermodynamical treatment of elementary quantum systems, most noticeably those far from equilibrium, strongly coupled and correlated to their environments \cite{Weimer_etal_2008, Hossein-Nejad_etal_2015, Alipour_etal_2016, Rivas_2020, Silva_Angelo_2021, Colla_Breuer_2021, Alipour_etal_2022, Ochoa_etal_2016, Carrega_etal_2016, Valente_etal_2018, Dou_etal_2018, Micadei_etal_2019, Strasberg_2019, Ali_Huang_Zhang_2020, Pyharanta_etal_2022, Bernardo_2021}. A related trend is to leave off the conception of work as energy exchanged with a classical agent, aiming toward the study of energy exchanges within a closed, autonomous ``quantum universe'' \cite{Hossein-Nejad_etal_2015, Alipour_etal_2016, Colla_Breuer_2021, Alipour_etal_2022, Valente_etal_2018, Micadei_etal_2019, Pyharanta_etal_2022}. Those features arguably set up the ultimate scope of quantum thermodynamics.

	In this realm, it is certainly disputable whether or not one should expect to be able to speak of internal energy of an open quantum system. Many authors have been assuming the positive answer, either explicitly, or implicitly through definitions of work and heat. In this broad sense, several distinct notions of internal energy are available in the literature, each based upon a particular set of physical principles and/or conceptual requirements \cite{Weimer_etal_2008, Hossein-Nejad_etal_2015, Alipour_etal_2016, Rivas_2020, Silva_Angelo_2021, Colla_Breuer_2021, Alipour_etal_2022, Valente_etal_2018}.

	Here, we addressed the question of whether a general, ``universal'' definition of internal energy for open quantum systems may be designed. As a first attempt to rigorously approach this issue, we established two basic properties that such a hypothetical definition might be expected to satisfy. The first, most elementary requirement was named \textit{consistency}: in a bipartite universe, the same definition should apply indistinctly to the two parts, in such a way that the sum total of the internal energies is equal to the internal energy of the whole, which is well-defined \textit{a priori}. The second property, accepted provisionally, is that the effective Hamiltonian of each open system, at a given time, would depend only on its quantum state (reduced density operator) and the time derivative of the latter; operationally, that would be a minimalist implementation of internal energy being a local quantity, in the sense of being determined only by the local state dynamics.
	
	As a setup for testing the hypotheses above, we considered a closed, bipartite quantum universe comprising two two-level systems (TLS's), evolving autonomously from a pure quantum state. By developing the mathematical structure behind those requirements and with the aid of numerical computation, we showed that they cannot be simultaneously satisfied in that context. Our result equally applies if in the same context we allow for time-dependent Hamiltonians, provided the individual components vary with time only in eigenvalue but not in eigenstates. We actively intend to depart our method from those approaches relying on results derived for particular interaction models; in fact, given the restriction to two TLS's, we have considered all possible interactions between them.
	
	An honest interpretation of a negative result should account for every underlying hypothesis. So inspired, the consequences of our result might be one or more of the following. \textit{(i)} Internal energy should depend on the second and/or higher derivatives of the density operator. That would be an interesting fact in its own right, particularly because the expressions would have to involve terms of second order either on the global density operator or on the Hamiltonian. \textit{(ii)} Internal energy should depend on something more than the local state and its derivatives, which indeed seems to be the case of many existing proposals \cite{Weimer_etal_2008, Hossein-Nejad_etal_2015, Rivas_2020, Colla_Breuer_2021}. In this case, one should either carefully argue why such a prescription should be considered local, or deal with the consequences of conceiving ``internal'' energy as a non-local attribute.  
As a matter of fact, recently, it has been shown \cite{Malavazi_Brito_2022} that a consistent definition of internal energy can be constructed from the system's reduced matrix knowledge, but that is only well-determined if the system-environment interaction is known, which is certainly a non-local property.
\textit{(iii)} Internal energy is not consistent, \ie{}, in a closed universe, the internal energies of the parts do not add up to that of the whole. Then the idea of \textit{internal energy} for open quantum systems probably should be abandoned, and consequently those of work and heat as well. Or simply \textit{(iv)} neither of the above happens, but the idea of internal energy just is not ``universal'' in the sense that it does not apply for a universe of two TLS's and/or in such a general regime as considered here. Then the task would be to determine ``where'' the ``borders'' come about -- \changed{\ie{}, whether the limiting factor for thermodynamical addressability is small environment, strong interactions, presence of (quantum) correlations, or absence of thermal equilibrium, keeping in mind that, combined or not, all these elements are becoming familiar in modern quantum thermodynamics literature.}

	Our result may also be connected with the notion of local passivity, which is relevant in the study of energy extraction from local subsystems. As put forward by Refs. \cite{Frey_Funo_Hotta_2014,Alhambra_etal_2019}, CP-local passivity is defined as a condition of a bipartite system in which the global system energy cannot be reduced through the application of any local (CPTP) map on one of its parts, a condition which was shown to hold in some relevant situations, as well as to be linked with entanglement. Back to our framework, if an internal energy law (IEL) could be designed satisfying strong locality and consistency, then the total energy of a bipartite system would \textit{always} (\textit{i. e.} for any pure global state) be prone to change by suitable manipulation of the reduced density operators (and their derivatives), which is arguably achievable through CPTP operations and could possibly leave room for energy extraction. Further, our negative result was crucially built on the tacit requirement of universality, meaning that an IEL with the conjectured properties should ``work'' for all possible pure global states, most of which are, of course, entangled. In this sense, our findings might reinforce the notion that entanglement constrains the possibility of local energy extraction (or even addition).

	Besides deriving the general constraint above, we considered a particular counter-example, inspired by the approach of \cite{Valente_etal_2018}, where internal energy is (implicitly) defined as depending only on the system's state and its derivative, and seems to be consistent in our sense, as can be verified easily. We showed that, despite our setup of two TLS's being distinct from that considered in the referred work -- in which one of them should be replaced by a boson bath --, it allows for the interaction model and initial condition of the latter to be ``emulated'', by means of an exchange interaction with no initial amplitude in the two-excitation subspace, which was taken as our ``control case''. Then we have seen that indeed those particular conditions yield an apparently consistent notion of internal energy in the control case, but that such ``consistency'' breaks down upon the slightest generalization of either initial state or interaction Hamiltonian, which is in accordance with our general constraint. In any case, it is noteworthy that those particular setups, either in the ``spin-boson'' or in the ``spin-spin'' case, allow for such particular modes of splitting of the universe energy among the parts.
	
	In synthesis, our simplest message is that, if a presumably universal notion of internal energy for open quantum systems is to be addressed, then it should be shown to be compatible with the quantification of the energy of closed systems -- and that this requirement by itself may impose strong constraints on the structure of the definition. We understand that our method may be generalized to systems of higher dimensionality, which should shed light on the possible implications of the constraint communicated here. We also believe that a deeper inquiry into the particular dynamical and/or informational properties of our so-called control case (or analogous setups) and their connection with thermodynamics may be fruitful.

\begin{acknowledgments}	
	L.R.T.N. acknowledges full financial support from São Paulo Research Foundation (FAPESP), grant \#{2021}/01365-9.  F.B. is supported by the Instituto Nacional de Ciência e Tecnologia de Informação Quântica (CNPq INCT-IQ 465469/2014-0). L.R.T.N. is warmly grateful for important discussions held with his fellows André Malavazi, João Inagaki, Lais Anjos, Matheus Fonseca, and Vitor Sena, throughout the development of this research -- with special thanks to Clara Vidor, Guilherme Zambon and Pedro Alcântara for critical reading of the early versions.% and discussions on the mathematical constructions, respectively.
\end{acknowledgments}

\appendix

\section{Internal energy for closed systems: $U = U(\rho, \dot{\rho})$}\label{sec:app:closed_system}

	The time evolution of a closed quantum system is dictated by the von-Neumann equation

	\begin{equation}
		\dot{\rho} = -i\left[ H, \rho\right],
	\end{equation}
where $\rho$ and $H$ are the system's density operator and Hamiltonian, respectively. Observe that the matrix elements of $\rho$ written in the eigenbasis of H are given by

	\begin{equation}\label{eq:closed_system:matrix_elements}
		\dot{\rho}_{nm} = -i\omega_{nm}\rho_{nm},
	\end{equation}
where $\omega_{nm}$ represent the system's spectral energy differences. The system's internal energy is found to be

	\begin{align}\label{eq:closed_system:intermediate}
		U &:= \avg{H} = \tr{\left(\rho H\right)} = \sum_{n}E_n\rho_{nn} \nonumber \\
		&= E_0\rho_{00} + \sum_{n\neq 0}E_n\rho_{nn} =
		E_0 + \sum_{n\neq 0}\omega_{n0}\rho_{nn},
	\end{align} 
where we have used the property $\tr{\rho}=1$ in the last step. Then, as long as $\rho$ does not commute with $H$ $\forall t$, using the relation \eqref{eq:closed_system:matrix_elements} in Equation \eqref{eq:closed_system:intermediate}, the internal energy of a closed quantum system is cast as the following functional:

	\begin{equation}
		U = E_0 + \sum_{n\neq 0}
		\im\left( \frac{\dot{\rho}_{0n}}{\rho_{0n}}\right)		
		\rho_{nn} = U(\rho, \dot{\rho}).
	\end{equation}

\section{Proof of Proposition \ref{prop:lin_system_method}}
\label{sec:app:proof_method}
	
    \hspace{-\parindent}\textit{Preliminaries. --} Let us represent an increment $\df x = (\df x_1,\hdots,\df x_n)^\top\in\mathbb{R}^n$ as a column vector. Let $x_0\in U$. We represent the derivative of $f$ at $x_0$,
    
        \begin{equation}
            \text{D}_{x_0}f := \left(\dpar{f}{x_1}{}(x_0), \hdots, \dpar{f}{x_n}{}(x_0)\right),
        \end{equation}
as a \textit{row} vector, such that the scalar $\df f = \left(\text{D}_{x_0}f\right) \cdot \df x$ is obtained as a matrix product. We adopt this convention consistently. For instance,
    
        \begin{equation}
            \left( \text{D}_{x_0}g\right)_{ij} := \dpar{g_i}{x_j}{}(x_0)
        \end{equation}
is an $m \times n$ matrix, representing a linear map from $\mathbb{R}^n$ to $\mathbb{R}^m$; indeed, the matrix product $\left( \text{D}_{x_0}g\right) \cdot \df x \in \mathbb{R}^m$ is a column vector giving the first-order increment in $g$ upon the displacement $\df x$ from $x_0$. 
    
    The elements above are given in the hypotheses. Now let $\hat{f}: g(U)\subset \mathbb{R}^m\to \mathbb{R}$ be \textit{some} smooth function. Like before, for $u_0\in g(U)$, 
    
        \begin{equation}
            \text{D}_{u_0}\hat{f} =  \left(\dpar{\hat{f}}{u_1}{}(u_0), \hdots, \dpar{\hat{f}}{u_m}{}(u_0)\right)
        \end{equation}
is a row vector.
        
    Consider the composition $\hat{f}\circ g :U\to \mathbb{R}$. Its derivative $\text{D}_{x_0}(\hat{f}\circ g)$ is a third row vector. It relates to the derivatives of $\hat{f}$ and $g$ as follows:
 
        \begin{align*}
            \left( \text{D}_{x_0}(\hat{f}\circ g) \right)_{j} 
            &= \dpar{(\hat{f}\circ g)}{x_j}{}(x_0)\\
            &= \sum_{k=1}^m
            \dpar{\hat{f}}{u_k}{}(g(x_0)) 
            \dpar{g_k}{x_j}{}(x_0)\\
            &= \sum_{k=1}^m
            \left(\text{D}_{g(x_0)}\hat{f}\right)_{k} 
            \left(\text{D}_{x_0}g\right)_{kj},
        \end{align*}
    \ie{}, we have the intuitive matrix relation
        
        \begin{equation}\label{eq:proof_method:eq1}
            \text{D}_{x_0}(\hat{f}\circ g) = \text{D}_{g(x_0)}\hat{f} \cdot \text{D}_{x_0}g.
        \end{equation}
    
    \hspace{-\parindent}\textit{Result}. -- By contradiction, assume that the system \eqref{eq:gen_res:lin_system_explicit} is not \textit{inconsistent in all of} $U$. Then choose an $x_0\in U$ such that it is consistent at $x_0$ and let  $\df \bar{x}\in\mathbb{R}^n$ be a solution of the system at that point. This is exactly to say that 
        
        \begin{equation}\label{eq:proof_method:eq2}
            \left( \text{D}_{x_0}g\right) \cdot \df \bar{x} = 0 \in \mathbb{R}^m
        \end{equation}
    and 
    
        \begin{equation}\label{eq:proof_method:eq3}
            \left(\text{D}_{x_0}f\right)\cdot \df \bar{x} = \delta f \neq 0.
        \end{equation}
        
    Now Equation \eqref{eq:proof_method:eq2} in \eqref{eq:proof_method:eq1} implies
    
        \begin{equation}\label{eq:proof_method:eq4}
            \left[\text{D}_{x_0}(\hat{f}\circ g)\right]\cdot \df \bar{x} = 0.
        \end{equation}

    But recall that $\hat{f}:g(U)\subset\mathbb{R}^m\to\mathbb{R}$ was arbitrary. Then, if we could choose $\hat{f}$ such that $\hat{f}\circ g = f$ around $x_0$ (hypothesis), we would immediately find a contradiction between Equations \eqref{eq:proof_method:eq3} and \eqref{eq:proof_method:eq4}. The conclusion is that $\hat{f}\circ g = f$ cannot be true in a vicinity of $x_0$ and therefore cannot be true in $U$. The claim is now proven.

\section{Proof of Theorem}
\label{sec:app:proof_applied}

		\hspace{-\parindent}\textit{First part. --} The first step is to translate the statement of Proposition \ref{prop:if_solvable} into a relation between functions defined on regions of $\mathbb{R}^n$. It will suffice to give explicit form to some correspondences that have already been defined throughout \S \ref{sec:formalism}.
		 
		First, given a configuration representation $X\in S$, the scalar $\avg{H}\in\mathbb{R}$ is obviously well defined. This defines a function from $S$ to $\mathbb{R}$ which we name $f$. (It is straightforward to write it explicitly, but there is no need to.) 
		
		A pair of represented 1-extended states, $(\sigma_1\sa,\sigma_1\sb)\in\Sigma_1^2\subset \mathbb{R}^{12}$, is also defined for every $X\in S$. To see that explicitly, it is most convenient to resort to the basis $\mathcal{N}$, but it should be clear that the procedure is well-defined (recall Remark \ref{remark:sigma_func_x}).  It defines a function $g:S\to\Sigma_1^2$.
		
		Now, assume Problem \ref{prob:weak} to be solvable. Then, there exists a $\mathcal{G} : \mathbb{\Sigma}_1\sa \times \mathbb{\Sigma}_1\sb \to \mathbb{R}$ such that, for every $\mathbf{X}\in\mathbb{S}$, $\mathcal{G}\left(\boldsymbol{\sigma}_1\sa (\mathbf{X}),\boldsymbol{\sigma}_1\sb (\mathbf{X}) \right) = \avg{H}(\mathbf{X})$ (Proposition \ref{prop:if_solvable}, Equation \ref{eq:gen_res:if_solvable}). We naturally introduce two compositions in order to derive a statement on functions between $S\subset \mathbb{R}^{19}$, $\Sigma_1^2\subset\mathbb{R}^{12}$, and $\mathbb{R}$. First, let $\mathbf{a}:S\to\mathbb{S}$ be the function assigning to every $X\in S$ the configuration that it represents (Remark \ref{remark:config_rep_not_injective}). Let also $
\mathbf{b}:\Sigma_1^2\to\mathbb{\Sigma}_1\sa\times\mathbb{\Sigma}_1\sb 
$ assign to every 12-list $(\sigma_1\sa,\sigma_1\sb)\in\Sigma_1^2$ a pair of 1-extended states, of $A$ and $B$, respectively; this map is well-defined (recall Remark \ref{remark:ex_state_rep_mappings}). When denoting the action of these two, we apply square brackets, for more clarity. Also, define $\hat{f} : \Sigma_1^2\to\mathbb{R}$ as the composition $\hat{f} := \mathcal{G} \circ \mathbf{b}$. This one assigns to every a 12-list of $\Sigma_1^2$ the corresponding average energy $\avg{H}$ (given by the IEL $\mathcal{E}$). 

\iffalse
Let also $r:\mathbb{\Sigma}_1\sa\times\mathbb{\Sigma}_1\sb\to\Sigma_1^2$ map every \textit{ordered pair} of 1-extended states of $A, B$ into their ordered ``full'' representation (this is well defined; recall Remark \ref{remark:ex_state_rep_mappings}).
\fi		

		Given $X \in S$, we are allowed to apply $g$, obtaining a 12-list representing a pair of 1-extended states, and then apply $\hat{f}$, obtaining a real number with the meaning of average energy. Rigorously, we then have, for every $X\in S$,

		\begin{align}
			\hat{f}\left( g(X) \right) &= 
			\mathcal{G}\left( \mathbf{b}[g(X)] \right) 
			&&\text{ by definition of } \hat{f} \nonumber \\
			&= \mathcal{G}\left( 
				\boldsymbol{\sigma}_1\sa(\mathbf{a}[X]),
				\boldsymbol{\sigma}_1\sb(\mathbf{a}[X]) 
			\right)
			&& \nonumber \\
			&= \avg{H}\left( \mathbf{a}[X]\right) 
			&&\text{ by Eq. \eqref{eq:gen_res:if_solvable} }\nonumber \\ 
			&= f(X)
			&&\text{ by definition of } f.
		\end{align}			
		
		The second passage stems from the fact, clear from the definitions, that $\mathbf{b}\left[g(X)\right] = \left( 
				\boldsymbol{\sigma}_1\sa(\mathbf{a}[X]),
				\boldsymbol{\sigma}_1\sb(\mathbf{a}[X]) 
			\right)$, identically. 
			
		The diagram of Fig. \ref{fig:prop3_proof_compositions} summarizes all the mappings used; the relationships between them should be clear.  

\ifPacks
		\begin{figure}
		\begin{tikzcd}
			X\in S 
				\arrow[mapsto, d, "\mathbf{a}"]
				\arrow[mapsto, ddr, "f", swap, bend left=20, color=red]
				\arrow[mapsto, ddd, "g", bend right=90, swap, color=red]				
			&
			{}
			\\
			\mathbf{X}\in\mathbb{S} 
				\arrow[mapsto, d]			
			&
			{}
			\\
			(\boldsymbol{\sigma}_1\sa,\boldsymbol{\sigma}_1\sb)\in\mathbb{\Sigma}_1\sa\times\mathbb{\Sigma}_1\sb 
				\arrow[mapsto, r, "\mathcal{G}", dashed] 
			&
			\avg{H}\in\mathbb{R}			
			\\
			(\sigma_1\sa,\sigma_1\sb)\in\Sigma_1^2
				\arrow[mapsto, u, "\mathbf{b}"]
				\arrow[mapsto, ur, "\hat{f}", dashed, swap, bend right=10, color=red] 
			&
			{}
		\end{tikzcd}
		\caption{
			Scheme of the mappings involved in the first part of the proof of the Theorem. Red arrows highlight the applications that are of interest in the remainder of the proof, illustrating the relationship between them: $\hat f \circ g = f$. Dashed arrows indicate the applications that only exist due to the hypothesis that the underlying IEL, $\mathcal{E}$, is consistent and strongly 1-local.  		
		}
		\label{fig:prop3_proof_compositions}
		\end{figure}
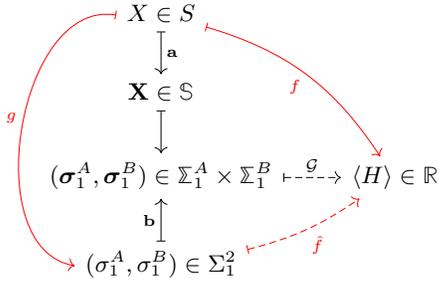 
\fi
		
		The first part of the proof is by now finished. It may seem that Proposition \ref{prop:lin_system_method} can be immediately applied, but not yet. By a matter of fact, $\int(S)\subset \mathbb{R}^{19}$ is not an open set; it is an embedded hypersurface, having dimensionality 18 (Remark \ref{remark:18dimhype}). This is an essential matter, since a naive application of the test suggested by Proposition \ref{prop:lin_system_method} when the domain $U$ of $f$ is an embedded hypersurface instead of an open region could yield ``false positive'' solutions $\df x$ such that $x_0 + \df x$ does not ``live'' in $U$ (more rigorously, such that $\df x\in\mathbb{R}^n$ does not lie in the tangent space of $U$). By dealing with this problem rigorously we will reach the intuitive conclusion that a proper constraint on the $R_k$ needs to be added ``by hand'', as in \eqref{eq:gen_result:system}. 

		\hspace{-\parindent}\textit{Second part. --} The crucial step is to put $\int(S)$ in one-to-one correspondence with an open set of $\mathbb{R}^{18}$, which is simply achieved by parameterizing the $R_k$ in a way that identically satisfies the constraint $\sum_kR_k^2=1$. This is achieved with hyperspherical coordinates: the set
		
		\begin{equation}
			D:= \left\lbrace
				(r, \alpha, \beta, \gamma) \in\mathbb{R}^4 :
				r\geqslant0; 0\leqslant \alpha, \beta \leqslant \pi; 0\leqslant \gamma < 2\pi	
			\right\rbrace
		\end{equation}					
parameterizes $\set{R_k}=\mathbb{R}^4$ by the prescription $\mathbf{c}:D\to\mathbb{R}^4$ described below:

		\begin{equation}
		\begin{aligned}
			R_1&=r\cos\alpha, \\
			R_2&=r\sin\alpha\cos\beta, \\
			R_3&=r\sin\alpha\sin\beta\cos\gamma, \\
			R_0&=r\sin\alpha\sin\beta\sin\gamma,
		\end{aligned}
		\end{equation}
whose inverse $\mathbf{h}:\mathbb{R}^4\to D$ is uniquely defined except for some pathological points which are not of our interest (analogous to the polar axis $z$ in the case of spherical coordinates in $\mathbb{R}^3$).
	
		Now, by applying $\mathbf{h}$ to the $R_k$ coordinates of $X \in \int S$ (which does not include any such pathological points) we find such set to be in one-to-one correspondence with
		
		\begin{widetext}
		\begin{equation}
			T = \left\lbrace
            \left(
                \alpha,\beta,\gamma;
                \left\lbrace \theta_k \right\rbrace_{k=0,\dots,3};
                \oma, \omb;
                \left\lbrace h_{jk}\right\rbrace_{j,k=x,y,z}
            \right)	
            %\right. \\ \left. 
            \in\mathbb{R}^{18} : 
            0 < \alpha, \beta, \gamma < \pi/2;
             0 < \theta_k<2\pi, \omj>0
			\right\rbrace.
		\end{equation}
		\end{widetext}
		
		We suggestively denote the correspondences just built as $\mathbf{\tilde{h}}: \int S \to T$ and $\mathbf{\tilde{c}} = \mathbf{\tilde{h}}^{-1}$.

		$T$ is an open set of $\mathbb{R}^{18}$ which parameterizes $\int S$; it remains to \textit{(a)} apply Proposition 2 directly to $T$, obtaining a statement on a linear system defined on $T$; and finally \textit{(b)} make the inverse movement, so as to build a corresponding statement on a linear system defined on $\int S$, namely \eqref{eq:gen_result:system}. 
		
		The conclusion of the first part was that there exists an $\hat{f}:\Sigma_1^2 \to \mathbb{R}$ such that, for every $X \in S$, $\hat{f}(g(X)) = f(X)$. Now, for every $Y\in T$, $\mathbf{\tilde{c}}(Y)\in S$ and therefore
		
		\begin{align}
			\hat{f}\left( (g\circ \mathbf{\tilde{c}}) (Y)\right) &:=
			\hat{f}\left( g( \mathbf{\tilde{c}}(Y) )\right) \nonumber \\
			&= f(\mathbf{\tilde{c}}(Y)) \nonumber \\
			&= (f \circ \mathbf{\tilde{c}})(Y).
		\end{align}
		
		Naming $\tilde{g}:=g\circ\mathbf{\tilde{c}}: T \to \Sigma_1^2$ and $\tilde{f}:=f\circ\mathbf{\tilde{c}} : T \to \mathbb{R}$, we have then found that, for every $Y \in T$, the function $\hat{f}$ already defined is such that $\hat{f}(\tilde{g}(Y)) = \tilde{f}(Y)$. Moreover it is simple to verify from the definitions that $g, f$ are differentiable functions in $\int S$ and thus by composition $\tilde{g}, \tilde{f}$ are differentiable in $T$ as well. Since $T\subset \mathbb{R}^{18}$ is an open region, Proposition \ref{prop:lin_system_method} applies and we conclude that, for every $Y_0 \in T$, the linear system
		
        \begin{equation}\label{eq:proof_applied:syst_aux}
            \left\lbrace
            \begin{aligned}
                \Df{Y_0} \tilde{g} \cdot \df Y &= 0 \\
                \Df{Y_0} \tilde{f} \cdot \df Y &= \delta E \neq 0
            \end{aligned}
            \right.
        \end{equation}
in the variable $\df Y \in \mathbb{R}^{18}$ is inconsistent. 

		Now we consider the other linear system

        \begin{equation}\label{eq:proof_applied:syst_final}
            \left\lbrace
            \begin{aligned}
                \Df{X_0} g \cdot \df X &= 0 \\
                \Df{X_0}\left(  \sum_{k=0}^3 R_{k}^2\right) \cdot \df X &= 0\\
                \Df{X_0} f \cdot \df X &= \delta E,\\
            \end{aligned}
            \right.
        \end{equation}
defined at $X_0\in\int S$, in the variable $\df X \in \mathbb{R}^{19}$. We wish to conclude that it is also inconsistent. It suffices to show that, given a hypothetical solution $\df \bar X \in \mathbb{R}^{19}$ of \eqref{eq:proof_applied:syst_final} at $X_0\in\int S$, we can build a solution $\df \bar Y\in\mathbb{R}^{18}$ of \eqref{eq:proof_applied:syst_aux} at $Y_0:=\mathbf{\tilde{h}}(X_0)\in T$. The procedure is clear. Under the conditions just given, we define, for $k=1,\dots,18$,

		\begin{equation}\label{eq:proof_applied:dy_def}
			\df \bar Y_k := \sum_{j=1}^{19} \dpar{
				\mathbf{\tilde{h}}_k			
			}{X_j}{}\left( X_0 \right) \df \bar X_j.
		\end{equation}
		
		It is intuitive to invert the relationship above and write $\df \bar X_j$ as a combination of $\df \bar Y_k$. First recall that $\mathbf{\tilde{h}}$ merely transforms the first 4 coordinates of $X$ to the three hyperspherical angles $\alpha,\beta,\gamma$. It follows directly that $\partial \mathbf{\tilde{h}}_k/\partial X_j = \delta_{k,j-1}$ for $k = 4,\dots,18$. For this range of $k$, thus, \eqref{eq:proof_applied:dy_def} yields $\df \bar Y_k = \df \bar X_{k+1}$, as expected. Now for $k=1,2,3$, it gives
		
		\begin{align}
			\df \bar Y_k &= \sum_{j=1}^{19} \dpar{
				\mathbf{\tilde{h}}_k			
			}{X_j}{}\left( X_0 \right) \df \bar X_j 
			\nonumber\\
			&= \sum_{j=1}^{4} \dpar{
				\mathbf{\tilde{h}}_k			
			}{X_j}{}\left( X_0 \right) \df \bar X_j 
			\nonumber\\
			&= \sum_{j=0}^{3} \dpar{
				\mathbf{h}_{k+1}			
			}{R_j}{}\left( \left\lbrace R_{\ell} \right\rbrace \right) \df \bar X_{j+1}^.
		\end{align}
	
		We basically used the fact that those $\mathbf{\tilde{h}}_k$ only depend on the first 4 entries of $X$, namely $\left\lbrace R_k\right\rbrace$, and act on them like $\mathbf{h}_k$, by definition. (Of course, we are denoting as $\left\lbrace R_{\ell} \right\rbrace $ the corresponding entries of $X_0$, specifically.)
		
		At this level it becomes clear that the relationship above can be inverted. In general, because $\mathbf{h}$ is invertible in $\left\lbrace R_{\ell}\right\rbrace$ (recall that $X_0\in\int S$), we would write the $\df R_j$ as a combination of $(\df r, \df \alpha, \df \beta, \df \gamma)$, the coefficients being the partial derivatives of $\mathbf{c}_j$. However, the corresponding value of $\df r = \df \mathbf{h}_1$ would vanish, since
		
		\begin{equation}
			\df r = \sum_{j=0}^3 \dpar{\mathbf{h}_1}{ R_j}{}(\left\lbrace R_\ell\right\rbrace)\df R_j  
		\end{equation}
while
$\mathbf{h}_1(\left\lbrace R_\ell\right\rbrace) = \sqrt{\textstyle\sum_{\ell=0}^3R_\ell ^2}$;
it then follows that the sum is proportional to $R_0\df R_0 + \dots + R_2\df R_2$, which is zero from the hypothesis that $\df \bar X$ is a solution of \eqref{eq:proof_applied:syst_final}. Therefore, we can write, for $j=1,...,4$,

		\begin{equation}
			\df \bar X_j = \sum_{k=1}^3 \dpar{\mathbf{\tilde{c}}_j}{Y_k}{}(Y_0)\df \bar Y_k.
		\end{equation}
		
		Now the sum above can be extended to $k=18$ since those $\mathbf{\tilde{c}}_j$ depend only on the first three entries of $Y$; and, for $j=5,\dots,19$, we already know the equation above to hold, since in this range $\partial\mathbf{\tilde{c}}_j/\partial Y_k=\delta_{j,k+1}$ and thus the equation above would give $\df \bar X_j = \df \bar Y_{j-1}$, which was already shown to be the case. Therefore we can write, for $j=1,\dots,19$,

		\begin{equation}\label{eq:proof_applied:dy_inv}
			\df \bar X_j = \sum_{k=1}^{18} \dpar{\mathbf{\tilde{c}}_j}{Y_k}{}(Y_0)\df \bar Y_k,
		\end{equation}		
which is the inverse of \eqref{eq:proof_applied:dy_def}. We emphasize that such an inversion, though intuitive, was not ``trivial'', since $\df \bar X$ and $\df \bar Y$ and  are in principle defined in vector spaces of different dimensionalities, namely $19$ and $18$, and therefore it was essential to use the very special constraint that $\df \bar X$ satisfies by definition. 
		
		With the aid of \eqref{eq:proof_applied:dy_inv}, it is straightforward to show that $\df \bar Y$, as defined in \eqref{eq:proof_applied:dy_def}, is a solution of the system \eqref{eq:proof_applied:syst_aux}. For instance, the left hand side of the first of  \eqref{eq:proof_applied:syst_aux} gives for this vector
		
		\begin{align}
			\df \tilde{g}(Y_0) 
			&:= \sum_{k=1}^{18}\dpar{\tilde{g}}{Y_k}{}(Y_0)\df \bar Y_k
			\nonumber \\
			&= \sum_{k=1}^{18}\dpar{(g\circ \mathbf{\tilde{c}})}{Y_k}{}(Y_0)\df \bar Y_k
			\nonumber \\
			&= \sum_{k=1}^{18}
				\left( 
					\sum_{j=1}^{19}\dpar{g}{X_j}{}\left(\mathbf{\tilde{c}}(Y_0)\right)\dpar{\mathbf{\tilde{c}}_j}{Y_k}{}(Y_0)				
				\right)			
			\df \bar Y_k
			\nonumber \\
			&= \sum_{j=1}^{19}\dpar{g}{X_j}{}\left(X_0\right)
			\sum_{k=1}^{18}\dpar{\mathbf{\tilde{c}}_j}{Y_k}{}(Y_0)
			\df \bar Y_k
			\nonumber \\
			&= \sum_{j=1}^{19}\dpar{g}{X_j}{}\left(X_0\right)
			\df \bar X_j 
			\nonumber \\
			&= 0,
		\end{align}
where in the last two steps we used, respectively, \eqref{eq:proof_applied:dy_inv} and the first of \eqref{eq:proof_applied:syst_final}. The same reasoning shows, with even smaller effort, that 

		\begin{equation}
			\df \tilde f (Y_0) = \delta E.
		\end{equation}	
		
		\hspace{-\parindent}\textit{Summary. --} We have just proven that, given a solution $\df \bar X\in\mathbb{R}^{19}$ of the linear system \eqref{eq:proof_applied:syst_final} defined at a point $X_0\in\int S$, there exists one solution $\df \bar Y \in \mathbb{R}^{18}$ of the system \eqref{eq:proof_applied:syst_aux} at a certain point $Y_0\in T$. But we had already shown that the latter cannot exist under the hypothesis that Problem \ref{prob:weak} is solvable. The conclusion is that, under such hypothesis, \eqref{eq:proof_applied:syst_final} is an unsolvable system for every $X_0\in \int S$. Equation \eqref{eq:gen_result:system} is obviously just another way of writing the same system. 
	
\section{Explicit calculations for \S \ref{subsec:counter_ex:excit_cons}}
\label{sec:app:counter_ex_explicit}

\hspace{-\parindent}\textit{Derivatives of local density matrix with $H^\text{num.}_{\Lambda,\Delta}$. --} Due to the particular form of $\ket{\psi}$ \eqref{eq:counter_ex:dynamics:universe_state}, the universe density matrix is 
    
        \begin{equation}
            \rho = \kb{\psi} =
            \begin{pmatrix}
                \absq{\psi_0} &
                \psi_0\psi_B^* &
                \psi_0\psi_A^* &
                0 \\
                \psi_0^*\psi_B &
                \absq{\psi_B} &
                \psi_A^*\psi_B &
                0 \\
                \psi_0^*\psi_A &
                \psi_A\psi_B^* &
                \absq{\psi_A} &
                0 \\
                0 &
                0 &
                0 &
                0 
            \end{pmatrix}
        \end{equation}
so that, with $\hint=H^\text{num.}_{\Lambda,\Delta}$ \eqref{eq:counter_ex:num_ham},  
		\begin{widetext}   
        \begin{equation}
            \hint \rho = 
            \begin{pmatrix}
                \Delta \absq{\psi_0} &
                \Delta \psi_0 \psi_B^* &
                \Delta \psi_0 \psi_A^* &
                0 \\
                \Lambda^*\psi_0^*\psi_A - \Delta \psi_0^* \psi_B &
                \Lambda^* \psi_A \psi_B^* - \Delta \absq{\psi_B} &
                \Lambda^*\absq{\psi_A} - \Delta\psi_A^*\psi_B & 0 \\ 
                \Lambda\psi_0^*\psi_B - \Delta \psi_0^* \psi_A &
                \Lambda \absq{\psi_B} - \Delta\psi_A \psi_B^*  &
                \Lambda\psi_A^*\psi_B - \Delta\absq{\psi_A} & 0 \\
                0&0&0&0
            \end{pmatrix}
        \end{equation}
and, from $\rho\hint = (\hint \rho)\hc$,
        
        \begin{equation}\label{eq:counter_ex:appendix:rho_hint}
        %\begin{split}
            \rho\hint =
            %\\
            \begin{pmatrix}
                \Delta \absq{\psi_0} &
                \Lambda\psi_0\psi_A^* - \Delta \psi_0 \psi_B^* &
                \Lambda^*\psi_0\psi_B^* - \Delta \psi_0 \psi_A^* &
                0\\
                \Delta \psi_0^* \psi_B &
                \Lambda \psi_A^* \psi_B - \Delta \absq{\psi_B} &
                \Lambda^* \absq{\psi_B} - \Delta\psi_A^* \psi_B  &
                0\\
                \Delta \psi_0^* \psi_A &
                \Lambda\absq{\psi_A} - \Delta\psi_A\psi_B^* &
                \Lambda^*\psi_A\psi_B^* - \Delta\absq{\psi_A} & 
                0\\                
                0&0&0&0
            \end{pmatrix}.
		%\end{split}
        \end{equation}
		\end{widetext}
    
	This gives the Hermitian matrix
    
        \begin{multline}
            -i\left[\hint,\rho\right] = \\
            \begin{pmatrix}
                0 &
                i\psi_0\left(\Lambda  \psi_A^* - 2\Delta  \psi_B^*\right) &
                -i\psi_0\left( 2\Delta \psi_A^* - \Lambda^* \psi_B^* \right) & 
                0 \\
                * &
                2\im{\left(\Lambda^*\psi_A\psi_B^*\right)} &
                -i\Lambda^*\left( \absq{\psi_A}-\absq{\psi_B}\right) &
                0 \\
                *&
                *&
                -2\im{\left(\Lambda^*\psi_A\psi_B^*\right)} &
                0\\
                *&*&*&0
            \end{pmatrix}.
        \end{multline}

	In turn, $H\sa+H\sb=\text{diag}(0, \omega\sb,\omega\sa,\omega\sa+\omega\sb)$ yields 
	           
        \begin{multline}
            -i\left[H\sa +H\sb, \rho\right] =\\
            \begin{pmatrix}
                0 &
                i\psi_0\omb\psi_B^*&
                i\psi_0\oma\psi_A^*&
                0\\
                *&0&
                i(\oma-\omb)\psi_A^*\psi_B&
                0\\
                *&*&0&0\\
                *&*&*&0
            \end{pmatrix}
        \end{multline}

and thus, since $\dot{\rho} = -i\left[H,\rho\right]$,

		\begin{widetext}
        \begin{equation}
            \dot{\rho} =
            \begin{pmatrix}
                0 &
                i\psi_0\left[
                    \Lambda \psi_A^*+
                    (\omb-2\Delta)\psi_B^*
                \right] &
                i\psi_0\left[
                    \Lambda^* \psi_B^*+
                    (\oma-2\Delta)\psi_A^*
                \right] &
                0 \\
                * &
                -2\im{\left(\Lambda\psi_A^*\psi_B\right)} &
                i\left[ 
                    (\oma-\omb)\psi_A^*\psi_B -
                    \Lambda^*\left(\absq{\psi_A}-\absq{\psi_B}\right)
                \right] &
                0 \\
                * & * &
                +2\im{\left(\Lambda\psi_A^*\psi_B\right)}&
                0\\
                *&*&*&0
            \end{pmatrix}.
        \end{equation}           
        \end{widetext}

	By partial-tracing the equation above, we immediately obtain the pair \eqref{eq:counter_ex:dynamics:local_derivatives}. 
	
\hspace{-\parindent}\textit{Average of $H$. --} From Equations \eqref{eq:counter_ex:dynamics:local_states} and the definitions of $H\sj$ we immediately have $\avg{H\sj}=\absq{\psi_j}\omj$. Now, the trace of \eqref{eq:counter_ex:appendix:rho_hint} gives $\avg{\hint}=2\re{\left(
                    \Lambda \psi_B \psi_A^*
                \right)}
            - (\absq{\psi_A}+\absq{\psi_B})\Delta
            + \absq{\psi_0}\Delta$, so that
            
        \begin{equation}
        \begin{aligned}
            \avg{H} &= 
            \absq{\psi_A}\oma 
            + \absq{\psi_B}\omb 
            + 2\re{\left(
                    \Lambda \psi_B \psi_A^*
                \right)} \\
            &- (\absq{\psi_A}+\absq{\psi_B})\Delta
            + \absq{\psi_0}\Delta
		\end{aligned}        
        \end{equation}
    which by normalization may be rewritten as \eqref{eq:counter_ex:consistency:mean_h}.

\bibliography{bibliography.bib}

\end{document}

%% file: shortcuts.tex
%\numberwithin{equation}{section}

%%% somehow general, useful commands

%highlights, abbreviations:

\newcommand{\red}[1]{{\color{red}#1}}
\newcommand{\blue}[1]{{\color{blue}#1}}
\newcommand{\warning}[1]{
\red{\textbf{\underline{#1}}}
}

\newcommand{\citneed}{
\textbf{\red{[citation needed!]}}
}

\newcommand{\cf}{\textit{cf.} }
\newcommand{\ie}[1]{\textit{i. e.}#1}
\newcommand{\eg}[1]{\textit{e. g.}#1}

\theoremstyle{plain}
\newtheorem{definition}{Definition}
\newtheorem{proposition}{Proposition}
\newtheorem{theorem}{Theorem}
\newtheorem*{theorem*}{Theorem}

\theoremstyle{definition}
\newtheorem{remark}{Remark}[definition]
\newtheorem{remarkprop}{Note}[proposition]
\newtheorem*{remarkthm*}{Note}

\newtheorem{prob}{Problem}

\newenvironment{openboxed}{
\hspace{-.6mm-\parindent}
\rule{\linewidth}{.2mm}
}{
\hspace{-.6mm-\parindent}
\rule{\linewidth}{.2mm}
}

\newenvironment{problem}{
\begin{openboxed}
\vspace{-8mm}
\begin{prob}
}{
\end{prob}
\vspace{-6mm}
\end{openboxed}
}

%maths:

\newcommand{\der}[3]{       %#3-th order derivative of #1 with respect to #2
\frac{\text{d}^{#3} #1}{\text{d} #2 ^{#3} } 
}  

\newcommand{\dpar}[3]{       %#3-th order partial derivative of #1 with respect to #2
\frac{\partial ^{#3} #1}{\partial #2 ^{#3} } 
}  

\newcommand{\df}{\text{d}}  %differential

\newcommand{\Df}[1]{	%derivative operator, point #1
\text{D}_{#1}
}

\newcommand{\dt}{ \df_t } %time derivative, lazy

\newcommand{\norm}[1]{  %vector norm (double bars)
\left \| #1 \right \|
}

\newcommand{\abs}[1]{       %absolute value
\left | #1 \right |
}

\newcommand{\absq}[1]{      %squared absolute value
\left | #1 \right |^2
}

\renewcommand{\set}[1]{     %just brackets
\left\{ #1 \right\}
}

\newcommand{\expp}[1]{      %exponential in parentheses
\exp{\left( #1 \right)}
}

\newcommand{\expb}[1]{      %exponential in brackets
\exp{\left[ #1 \right]}
}

\newcommand{\expcb}[1]{      %exponential in curly brackets
\exp{\left\{ #1 \right\}}
}

\newcommand{\im}{ %imaginary part
\, \text{Im} \,
}

\newcommand{\re}{  %real part
\, \text{Re} \,
}

\newcommand{\avg}[1]{       %mean value
\braket{#1}
}

%%%more specific commands (quantum mechanics)

\newcommand{\kb}[1]{    %ketbra of #1 
\ket{#1}\bra{#1}
}

\newcommand{\kbb}[2]{   %ket #1 then bra #2
\ket{#1}\bra{#2}
}

\newcommand{\hc}{^\dagger}  %just an overscript dagger (hermitean conjugate)

\newcommand{\tr}{       %trace
\text{Tr} 
}

\newcommand{\trp}[1]{   %partial trace in the space indicated by #1
\text{Tr}_{\text{#1}}
}

\newcommand{\her}{	% set of Hermitean operators
\mathcal{L}_H
}
% very very specific commands

\newcommand{\hintantigo}{H^\text{int.}
}

\newcommand{\hint}{H^I
}

\newcommand{\sj}{^{(j)}}

\newcommand{\sa}{^A}
\renewcommand{\sb}{^B}

\newcommand{\omj}{\omega \sj}

\newcommand{\oma}{\omega \sa}
\newcommand{\omb}{\omega \sb}

\newcommand{\ef}{\Tilde}

\newcommand{\atx}[1]{
 \left. #1 \right| _{\mathbf{X}}
}

\renewcommand{\int}{
\text{int}
}

%% file: main.bbl
\begin{thebibliography}{39}%
\makeatletter
\providecommand \@ifxundefined [1]{%
 \@ifx{#1\undefined}
}%
\providecommand \@ifnum [1]{%
 \ifnum #1\expandafter \@firstoftwo
 \else \expandafter \@secondoftwo
 \fi
}%
\providecommand \@ifx [1]{%
 \ifx #1\expandafter \@firstoftwo
 \else \expandafter \@secondoftwo
 \fi
}%
\providecommand \natexlab [1]{#1}%
\providecommand \enquote  [1]{``#1''}%
\providecommand \bibnamefont  [1]{#1}%
\providecommand \bibfnamefont [1]{#1}%
\providecommand \citenamefont [1]{#1}%
\providecommand \href@noop [0]{\@secondoftwo}%
\providecommand \href [0]{\begingroup \@sanitize@url \@href}%
\providecommand \@href[1]{\@@startlink{#1}\@@href}%
\providecommand \@@href[1]{\endgroup#1\@@endlink}%
\providecommand \@sanitize@url [0]{\catcode `\\12\catcode `\$12\catcode
  `\&12\catcode `\#12\catcode `\^12\catcode `\_12\catcode `\%12\relax}%
\providecommand \@@startlink[1]{}%
\providecommand \@@endlink[0]{}%
\providecommand \url  [0]{\begingroup\@sanitize@url \@url }%
\providecommand \@url [1]{\endgroup\@href {#1}{\urlprefix }}%
\providecommand \urlprefix  [0]{URL }%
\providecommand \Eprint [0]{\href }%
\providecommand \doibase [0]{https://doi.org/}%
\providecommand \selectlanguage [0]{\@gobble}%
\providecommand \bibinfo  [0]{\@secondoftwo}%
\providecommand \bibfield  [0]{\@secondoftwo}%
\providecommand \translation [1]{[#1]}%
\providecommand \BibitemOpen [0]{}%
\providecommand \bibitemStop [0]{}%
\providecommand \bibitemNoStop [0]{.\EOS\space}%
\providecommand \EOS [0]{\spacefactor3000\relax}%
\providecommand \BibitemShut  [1]{\csname bibitem#1\endcsname}%
\let\auto@bib@innerbib\@empty
%</preamble>
\bibitem [{\citenamefont {Alicki}(1979)}]{Alicki_1979}%
  \BibitemOpen
  \bibfield  {author} {\bibinfo {author} {\bibfnamefont {R.}~\bibnamefont
  {Alicki}},\ }\bibfield  {title} {\bibinfo {title} {The quantum open system as
  a model of the heat engine},\ }\href
  {https://doi.org/10.1088/0305-4470/12/5/007} {\bibfield  {journal} {\bibinfo
  {journal} {Journal of Physics A: Mathematical and General}\ }\textbf
  {\bibinfo {volume} {12}},\ \bibinfo {pages} {L103–L107} (\bibinfo {year}
  {1979})}\BibitemShut {NoStop}%
\bibitem [{\citenamefont {Weimer}\ \emph {et~al.}(2008)\citenamefont {Weimer},
  \citenamefont {Henrich}, \citenamefont {Rempp}, \citenamefont {Schröder},\
  and\ \citenamefont {Mahler}}]{Weimer_etal_2008}%
  \BibitemOpen
  \bibfield  {author} {\bibinfo {author} {\bibfnamefont {H.}~\bibnamefont
  {Weimer}}, \bibinfo {author} {\bibfnamefont {M.~J.}\ \bibnamefont {Henrich}},
  \bibinfo {author} {\bibfnamefont {F.}~\bibnamefont {Rempp}}, \bibinfo
  {author} {\bibfnamefont {H.}~\bibnamefont {Schröder}},\ and\ \bibinfo
  {author} {\bibfnamefont {G.}~\bibnamefont {Mahler}},\ }\bibfield  {title}
  {\bibinfo {title} {Local effective dynamics of quantum systems: A generalized
  approach to work and heat},\ }\href
  {https://doi.org/10.1209/0295-5075/83/30008} {\bibfield  {journal} {\bibinfo
  {journal} {EPL (Europhysics Letters)}\ }\textbf {\bibinfo {volume} {83}},\
  \bibinfo {pages} {30008} (\bibinfo {year} {2008})}\BibitemShut {NoStop}%
\bibitem [{\citenamefont {Hossein-Nejad}\ \emph {et~al.}(2015)\citenamefont
  {Hossein-Nejad}, \citenamefont {O’Reilly},\ and\ \citenamefont
  {Olaya-Castro}}]{Hossein-Nejad_etal_2015}%
  \BibitemOpen
  \bibfield  {author} {\bibinfo {author} {\bibfnamefont {H.}~\bibnamefont
  {Hossein-Nejad}}, \bibinfo {author} {\bibfnamefont {E.~J.}\ \bibnamefont
  {O’Reilly}},\ and\ \bibinfo {author} {\bibfnamefont {A.}~\bibnamefont
  {Olaya-Castro}},\ }\bibfield  {title} {\bibinfo {title} {Work, heat and
  entropy production in bipartite quantum systems},\ }\href
  {https://doi.org/10.1088/1367-2630/17/7/075014} {\bibfield  {journal}
  {\bibinfo  {journal} {New Journal of Physics}\ }\textbf {\bibinfo {volume}
  {17}},\ \bibinfo {pages} {075014} (\bibinfo {year} {2015})}\BibitemShut
  {NoStop}%
\bibitem [{\citenamefont {Alipour}\ \emph {et~al.}(2016)\citenamefont
  {Alipour}, \citenamefont {Benatti}, \citenamefont {Bakhshinezhad},
  \citenamefont {Afsary}, \citenamefont {Marcantoni},\ and\ \citenamefont
  {Rezakhani}}]{Alipour_etal_2016}%
  \BibitemOpen
  \bibfield  {author} {\bibinfo {author} {\bibfnamefont {S.}~\bibnamefont
  {Alipour}}, \bibinfo {author} {\bibfnamefont {F.}~\bibnamefont {Benatti}},
  \bibinfo {author} {\bibfnamefont {F.}~\bibnamefont {Bakhshinezhad}}, \bibinfo
  {author} {\bibfnamefont {M.}~\bibnamefont {Afsary}}, \bibinfo {author}
  {\bibfnamefont {S.}~\bibnamefont {Marcantoni}},\ and\ \bibinfo {author}
  {\bibfnamefont {A.~T.}\ \bibnamefont {Rezakhani}},\ }\bibfield  {title}
  {\bibinfo {title} {Correlations in quantum thermodynamics: Heat, work, and
  entropy production},\ }\href {https://doi.org/10.1038/srep35568} {\bibfield
  {journal} {\bibinfo  {journal} {Scientific Reports}\ }\textbf {\bibinfo
  {volume} {6}},\ \bibinfo {pages} {35568} (\bibinfo {year}
  {2016})}\BibitemShut {NoStop}%
\bibitem [{\citenamefont {Rivas}(2020)}]{Rivas_2020}%
  \BibitemOpen
  \bibfield  {author} {\bibinfo {author} {\bibfnamefont {A.}~\bibnamefont
  {Rivas}},\ }\bibfield  {title} {\bibinfo {title} {Strong coupling
  thermodynamics of open quantum systems},\ }\href
  {https://doi.org/10.1103/PhysRevLett.124.160601} {\bibfield  {journal}
  {\bibinfo  {journal} {Physical Review Letters}\ }\textbf {\bibinfo {volume}
  {124}},\ \bibinfo {pages} {160601} (\bibinfo {year} {2020})}\BibitemShut
  {NoStop}%
\bibitem [{\citenamefont {Silva}\ and\ \citenamefont
  {Angelo}(2021)}]{Silva_Angelo_2021}%
  \BibitemOpen
  \bibfield  {author} {\bibinfo {author} {\bibfnamefont {T.~A. B.~P.}\
  \bibnamefont {Silva}}\ and\ \bibinfo {author} {\bibfnamefont {R.~M.}\
  \bibnamefont {Angelo}},\ }\bibfield  {title} {\bibinfo {title} {Quantum
  mechanical work},\ }\href {https://doi.org/10.1103/PhysRevA.104.042215}
  {\bibfield  {journal} {\bibinfo  {journal} {Phys. Rev. A}\ }\textbf {\bibinfo
  {volume} {104}},\ \bibinfo {pages} {042215} (\bibinfo {year}
  {2021})}\BibitemShut {NoStop}%
\bibitem [{\citenamefont {Colla}\ and\ \citenamefont
  {Breuer}(2021)}]{Colla_Breuer_2021}%
  \BibitemOpen
  \bibfield  {author} {\bibinfo {author} {\bibfnamefont {A.}~\bibnamefont
  {Colla}}\ and\ \bibinfo {author} {\bibfnamefont {H.-P.}\ \bibnamefont
  {Breuer}},\ }\bibfield  {title} {\bibinfo {title} {Exact open system approach
  to strong coupling quantum thermodynamics},\ }\href
  {http://arxiv.org/abs/2109.11893} {\bibfield  {journal} {\bibinfo  {journal}
  {arXiv:2109.11893 [quant-ph]}\ } (\bibinfo {year} {2021})},\ \bibinfo {note}
  {arXiv: 2109.11893}\BibitemShut {NoStop}%
\bibitem [{\citenamefont {Alipour}\ \emph {et~al.}(2022)\citenamefont
  {Alipour}, \citenamefont {Rezakhani}, \citenamefont {Chenu}, \citenamefont
  {del Campo},\ and\ \citenamefont {Ala-Nissila}}]{Alipour_etal_2022}%
  \BibitemOpen
  \bibfield  {author} {\bibinfo {author} {\bibfnamefont {S.}~\bibnamefont
  {Alipour}}, \bibinfo {author} {\bibfnamefont {A.~T.}\ \bibnamefont
  {Rezakhani}}, \bibinfo {author} {\bibfnamefont {A.}~\bibnamefont {Chenu}},
  \bibinfo {author} {\bibfnamefont {A.}~\bibnamefont {del Campo}},\ and\
  \bibinfo {author} {\bibfnamefont {T.}~\bibnamefont {Ala-Nissila}},\
  }\bibfield  {title} {\bibinfo {title} {Entropy-based formulation of
  thermodynamics in arbitrary quantum evolution},\ }\href
  {https://doi.org/10.1103/PhysRevA.105.L040201} {\bibfield  {journal}
  {\bibinfo  {journal} {Physical Review A}\ }\textbf {\bibinfo {volume}
  {105}},\ \bibinfo {pages} {L040201} (\bibinfo {year} {2022})}\BibitemShut
  {NoStop}%
\bibitem [{\citenamefont {Ochoa}\ \emph {et~al.}(2016)\citenamefont {Ochoa},
  \citenamefont {Bruch},\ and\ \citenamefont {Nitzan}}]{Ochoa_etal_2016}%
  \BibitemOpen
  \bibfield  {author} {\bibinfo {author} {\bibfnamefont {M.~A.}\ \bibnamefont
  {Ochoa}}, \bibinfo {author} {\bibfnamefont {A.}~\bibnamefont {Bruch}},\ and\
  \bibinfo {author} {\bibfnamefont {A.}~\bibnamefont {Nitzan}},\ }\bibfield
  {title} {\bibinfo {title} {Energy distribution and local fluctuations in
  strongly coupled open quantum systems: The extended resonant level model},\
  }\bibfield  {journal} {\bibinfo  {journal} {Physical Review B}\ }\textbf
  {\bibinfo {volume} {94}},\ \href {https://doi.org/10.1103/physrevb.94.035420}
  {10.1103/physrevb.94.035420} (\bibinfo {year} {2016})\BibitemShut {NoStop}%
\bibitem [{\citenamefont {Dou}\ \emph {et~al.}(2018)\citenamefont {Dou},
  \citenamefont {Ochoa}, \citenamefont {Nitzan},\ and\ \citenamefont
  {Subotnik}}]{Dou_etal_2018}%
  \BibitemOpen
  \bibfield  {author} {\bibinfo {author} {\bibfnamefont {W.}~\bibnamefont
  {Dou}}, \bibinfo {author} {\bibfnamefont {M.~A.}\ \bibnamefont {Ochoa}},
  \bibinfo {author} {\bibfnamefont {A.}~\bibnamefont {Nitzan}},\ and\ \bibinfo
  {author} {\bibfnamefont {J.~E.}\ \bibnamefont {Subotnik}},\ }\bibfield
  {title} {\bibinfo {title} {Universal approach to quantum thermodynamics in
  the strong coupling regime},\ }\href
  {https://doi.org/10.1103/PhysRevB.98.134306} {\bibfield  {journal} {\bibinfo
  {journal} {Physical Review B}\ }\textbf {\bibinfo {volume} {98}},\ \bibinfo
  {pages} {134306} (\bibinfo {year} {2018})}\BibitemShut {NoStop}%
\bibitem [{\citenamefont {Valente}\ \emph {et~al.}(2018)\citenamefont
  {Valente}, \citenamefont {Brito}, \citenamefont {Ferreira},\ and\
  \citenamefont {Werlang}}]{Valente_etal_2018}%
  \BibitemOpen
  \bibfield  {author} {\bibinfo {author} {\bibfnamefont {D.}~\bibnamefont
  {Valente}}, \bibinfo {author} {\bibfnamefont {F.}~\bibnamefont {Brito}},
  \bibinfo {author} {\bibfnamefont {R.}~\bibnamefont {Ferreira}},\ and\
  \bibinfo {author} {\bibfnamefont {T.}~\bibnamefont {Werlang}},\ }\bibfield
  {title} {\bibinfo {title} {Work on a quantum dipole by a single-photon
  pulse},\ }\href {https://doi.org/10.1364/OL.43.002644} {\bibfield  {journal}
  {\bibinfo  {journal} {Optics Letters}\ }\textbf {\bibinfo {volume} {43}},\
  \bibinfo {pages} {2644} (\bibinfo {year} {2018})}\BibitemShut {NoStop}%
\bibitem [{\citenamefont {Pyhäranta}\ \emph {et~al.}(2022)\citenamefont
  {Pyhäranta}, \citenamefont {Alipour}, \citenamefont {Rezakhani},\ and\
  \citenamefont {Ala-Nissila}}]{Pyharanta_etal_2022}%
  \BibitemOpen
  \bibfield  {author} {\bibinfo {author} {\bibfnamefont {T.}~\bibnamefont
  {Pyhäranta}}, \bibinfo {author} {\bibfnamefont {S.}~\bibnamefont {Alipour}},
  \bibinfo {author} {\bibfnamefont {A.~T.}\ \bibnamefont {Rezakhani}},\ and\
  \bibinfo {author} {\bibfnamefont {T.}~\bibnamefont {Ala-Nissila}},\
  }\bibfield  {title} {\bibinfo {title} {Correlation-enabled energy exchange in
  quantum systems without external driving},\ }\href
  {https://doi.org/10.1103/PhysRevA.105.022204} {\bibfield  {journal} {\bibinfo
   {journal} {Physical Review A}\ }\textbf {\bibinfo {volume} {105}},\ \bibinfo
  {pages} {022204} (\bibinfo {year} {2022})}\BibitemShut {NoStop}%
\bibitem [{\citenamefont {Strasberg}\ and\ \citenamefont
  {Esposito}(2019)}]{Strasberg_2019}%
  \BibitemOpen
  \bibfield  {author} {\bibinfo {author} {\bibfnamefont {P.}~\bibnamefont
  {Strasberg}}\ and\ \bibinfo {author} {\bibfnamefont {M.}~\bibnamefont
  {Esposito}},\ }\bibfield  {title} {\bibinfo {title} {Non-markovianity and
  negative entropy production rates},\ }\bibfield  {journal} {\bibinfo
  {journal} {Physical Review E}\ }\textbf {\bibinfo {volume} {99}},\ \href
  {https://doi.org/10.1103/physreve.99.012120} {10.1103/physreve.99.012120}
  (\bibinfo {year} {2019})\BibitemShut {NoStop}%
\bibitem [{\citenamefont {Carrega}\ \emph {et~al.}(2016)\citenamefont
  {Carrega}, \citenamefont {Solinas}, \citenamefont {Sassetti},\ and\
  \citenamefont {Weiss}}]{Carrega_etal_2016}%
  \BibitemOpen
  \bibfield  {author} {\bibinfo {author} {\bibfnamefont {M.}~\bibnamefont
  {Carrega}}, \bibinfo {author} {\bibfnamefont {P.}~\bibnamefont {Solinas}},
  \bibinfo {author} {\bibfnamefont {M.}~\bibnamefont {Sassetti}},\ and\
  \bibinfo {author} {\bibfnamefont {U.}~\bibnamefont {Weiss}},\ }\bibfield
  {title} {\bibinfo {title} {Energy exchange in driven open quantum systems at
  strong coupling},\ }\href {https://doi.org/10.1103/PhysRevLett.116.240403}
  {\bibfield  {journal} {\bibinfo  {journal} {Physical Review Letters}\
  }\textbf {\bibinfo {volume} {116}},\ \bibinfo {pages} {240403} (\bibinfo
  {year} {2016})}\BibitemShut {NoStop}%
\bibitem [{\citenamefont {Micadei}\ \emph {et~al.}(2019)\citenamefont
  {Micadei}, \citenamefont {Peterson}, \citenamefont {Souza}, \citenamefont
  {Sarthour}, \citenamefont {Oliveira}, \citenamefont {Landi}, \citenamefont
  {Batalhão}, \citenamefont {Serra},\ and\ \citenamefont
  {Lutz}}]{Micadei_etal_2019}%
  \BibitemOpen
  \bibfield  {author} {\bibinfo {author} {\bibfnamefont {K.}~\bibnamefont
  {Micadei}}, \bibinfo {author} {\bibfnamefont {J.~P.~S.}\ \bibnamefont
  {Peterson}}, \bibinfo {author} {\bibfnamefont {A.~M.}\ \bibnamefont {Souza}},
  \bibinfo {author} {\bibfnamefont {R.~S.}\ \bibnamefont {Sarthour}}, \bibinfo
  {author} {\bibfnamefont {I.~S.}\ \bibnamefont {Oliveira}}, \bibinfo {author}
  {\bibfnamefont {G.~T.}\ \bibnamefont {Landi}}, \bibinfo {author}
  {\bibfnamefont {T.~B.}\ \bibnamefont {Batalhão}}, \bibinfo {author}
  {\bibfnamefont {R.~M.}\ \bibnamefont {Serra}},\ and\ \bibinfo {author}
  {\bibfnamefont {E.}~\bibnamefont {Lutz}},\ }\bibfield  {title} {\bibinfo
  {title} {Reversing the direction of heat flow using quantum correlations},\
  }\href {https://doi.org/10.1038/s41467-019-10333-7} {\bibfield  {journal}
  {\bibinfo  {journal} {Nature Communications}\ }\textbf {\bibinfo {volume}
  {10}},\ \bibinfo {pages} {2456} (\bibinfo {year} {2019})}\BibitemShut
  {NoStop}%
\bibitem [{\citenamefont {Ali}\ \emph {et~al.}(2020)\citenamefont {Ali},
  \citenamefont {Huang},\ and\ \citenamefont {Zhang}}]{Ali_Huang_Zhang_2020}%
  \BibitemOpen
  \bibfield  {author} {\bibinfo {author} {\bibfnamefont {M.~M.}\ \bibnamefont
  {Ali}}, \bibinfo {author} {\bibfnamefont {W.-M.}\ \bibnamefont {Huang}},\
  and\ \bibinfo {author} {\bibfnamefont {W.-M.}\ \bibnamefont {Zhang}},\
  }\bibfield  {title} {\bibinfo {title} {Quantum thermodynamics of single
  particle systems},\ }\href {https://doi.org/10.1038/s41598-020-70450-y}
  {\bibfield  {journal} {\bibinfo  {journal} {Scientific Reports}\ }\textbf
  {\bibinfo {volume} {10}},\ \bibinfo {pages} {13500} (\bibinfo {year}
  {2020})}\BibitemShut {NoStop}%
\bibitem [{\citenamefont {de~Lima~Bernardo}(2021)}]{Bernardo_2021}%
  \BibitemOpen
  \bibfield  {author} {\bibinfo {author} {\bibfnamefont {B.}~\bibnamefont
  {de~Lima~Bernardo}},\ }\bibfield  {title} {\bibinfo {title} {Relating heat
  and entanglement in strong-coupling thermodynamics},\ }\bibfield  {journal}
  {\bibinfo  {journal} {Physical Review E}\ }\textbf {\bibinfo {volume}
  {104}},\ \href {https://doi.org/10.1103/physreve.104.044111}
  {10.1103/physreve.104.044111} (\bibinfo {year} {2021})\BibitemShut {NoStop}%
\bibitem [{\citenamefont {Welton}(1948)}]{Welton_1948}%
  \BibitemOpen
  \bibfield  {author} {\bibinfo {author} {\bibfnamefont {T.~A.}\ \bibnamefont
  {Welton}},\ }\bibfield  {title} {\bibinfo {title} {Some observable effects of
  the quantum-mechanical fluctuations of the electromagnetic field},\ }\href
  {https://doi.org/10.1103/PhysRev.74.1157} {\bibfield  {journal} {\bibinfo
  {journal} {Phys. Rev.}\ }\textbf {\bibinfo {volume} {74}},\ \bibinfo {pages}
  {1157} (\bibinfo {year} {1948})}\BibitemShut {NoStop}%
\bibitem [{\citenamefont {Breuer}\ and\ \citenamefont
  {Petruccione}(2002)}]{Breuer_Petruccione_2002}%
  \BibitemOpen
  \bibfield  {author} {\bibinfo {author} {\bibfnamefont {H.-P.}\ \bibnamefont
  {Breuer}}\ and\ \bibinfo {author} {\bibfnamefont {F.}~\bibnamefont
  {Petruccione}},\ }\href@noop {} {\emph {\bibinfo {title} {The theory of open
  quantum systems}}}\ (\bibinfo  {publisher} {Oxford University Press},\
  \bibinfo {year} {2002})\BibitemShut {NoStop}%
\bibitem [{\citenamefont {Vacchini}\ and\ \citenamefont
  {Breuer}(2010)}]{Vacchini_Breuer_2010}%
  \BibitemOpen
  \bibfield  {author} {\bibinfo {author} {\bibfnamefont {B.}~\bibnamefont
  {Vacchini}}\ and\ \bibinfo {author} {\bibfnamefont {H.-P.}\ \bibnamefont
  {Breuer}},\ }\bibfield  {title} {\bibinfo {title} {Exact master equations for
  the non-markovian decay of a qubit},\ }\href
  {https://doi.org/10.1103/PhysRevA.81.042103} {\bibfield  {journal} {\bibinfo
  {journal} {Physical Review A}\ }\textbf {\bibinfo {volume} {81}},\ \bibinfo
  {pages} {042103} (\bibinfo {year} {2010})}\BibitemShut {NoStop}%
\bibitem [{\citenamefont {Rivas}\ and\ \citenamefont
  {Huelga}(2012)}]{Rivas_Huelga_2012}%
  \BibitemOpen
  \bibfield  {author} {\bibinfo {author} {\bibfnamefont {A.}~\bibnamefont
  {Rivas}}\ and\ \bibinfo {author} {\bibfnamefont {S.~F.}\ \bibnamefont
  {Huelga}},\ }\href {https://doi.org/10.1007/978-3-642-23354-8} {\emph
  {\bibinfo {title} {Open Quantum Systems}}},\ SpringerBriefs in Physics\
  (\bibinfo  {publisher} {Springer Berlin Heidelberg},\ \bibinfo {year}
  {2012})\BibitemShut {NoStop}%
\bibitem [{\citenamefont {de~Vega}\ and\ \citenamefont
  {Alonso}(2017)}]{de_Vega_Alonso_2017}%
  \BibitemOpen
  \bibfield  {author} {\bibinfo {author} {\bibfnamefont {I.}~\bibnamefont
  {de~Vega}}\ and\ \bibinfo {author} {\bibfnamefont {D.}~\bibnamefont
  {Alonso}},\ }\bibfield  {title} {\bibinfo {title} {Dynamics of non-markovian
  open quantum systems},\ }\href {https://doi.org/10.1103/RevModPhys.89.015001}
  {\bibfield  {journal} {\bibinfo  {journal} {Reviews of Modern Physics}\
  }\textbf {\bibinfo {volume} {89}},\ \bibinfo {pages} {015001} (\bibinfo
  {year} {2017})}\BibitemShut {NoStop}%
\bibitem [{\citenamefont {Zhang}(2019)}]{Zhang_2019}%
  \BibitemOpen
  \bibfield  {author} {\bibinfo {author} {\bibfnamefont {W.-M.}\ \bibnamefont
  {Zhang}},\ }\bibfield  {title} {\bibinfo {title} {Exact master equation and
  general non-markovian dynamics in open quantum systems},\ }\href
  {https://doi.org/10.1140/epjst/e2018-800047-4} {\bibfield  {journal}
  {\bibinfo  {journal} {The European Physical Journal Special Topics}\ }\textbf
  {\bibinfo {volume} {227}},\ \bibinfo {pages} {1849} (\bibinfo {year}
  {2019})}\BibitemShut {NoStop}%
\bibitem [{\citenamefont {Callen}(1985)}]{Callen_1985}%
  \BibitemOpen
  \bibfield  {author} {\bibinfo {author} {\bibfnamefont {H.~B.}\ \bibnamefont
  {Callen}},\ }\href@noop {} {\emph {\bibinfo {title} {Thermodynamics and an
  introduction to thermostatistics}}},\ \bibinfo {edition} {2nd}\ ed.\
  (\bibinfo  {publisher} {Wiley},\ \bibinfo {address} {New York},\ \bibinfo
  {year} {1985})\BibitemShut {NoStop}%
\bibitem [{\citenamefont {Fermi}(1937)}]{Fermi_1937}%
  \BibitemOpen
  \bibfield  {author} {\bibinfo {author} {\bibfnamefont {E.}~\bibnamefont
  {Fermi}},\ }\href@noop {} {\emph {\bibinfo {title} {Thermodynamics}}}\
  (\bibinfo  {publisher} {Prentice-Hall},\ \bibinfo {address} {New York},\
  \bibinfo {year} {1937})\BibitemShut {NoStop}%
\bibitem [{\citenamefont {Milz}\ \emph {et~al.}(2017)\citenamefont {Milz},
  \citenamefont {Pollock},\ and\ \citenamefont {Modi}}]{Milz_etal_2017}%
  \BibitemOpen
  \bibfield  {author} {\bibinfo {author} {\bibfnamefont {S.}~\bibnamefont
  {Milz}}, \bibinfo {author} {\bibfnamefont {F.~A.}\ \bibnamefont {Pollock}},\
  and\ \bibinfo {author} {\bibfnamefont {K.}~\bibnamefont {Modi}},\ }\bibfield
  {title} {\bibinfo {title} {An introduction to operational quantum dynamics},\
  }\href {https://doi.org/10.1142/S1230161217400169} {\bibfield  {journal}
  {\bibinfo  {journal} {Open Systems \& Information Dynamics}\ }\textbf
  {\bibinfo {volume} {24}},\ \bibinfo {pages} {1740016} (\bibinfo {year}
  {2017})}\BibitemShut {NoStop}%
\bibitem [{\citenamefont {Alipour}\ \emph {et~al.}(2020)\citenamefont
  {Alipour}, \citenamefont {Rezakhani}, \citenamefont {Babu}, \citenamefont
  {M\o{}lmer}, \citenamefont {M\"ott\"onen},\ and\ \citenamefont
  {Ala-Nissila}}]{Alipour_etal_2020}%
  \BibitemOpen
  \bibfield  {author} {\bibinfo {author} {\bibfnamefont {S.}~\bibnamefont
  {Alipour}}, \bibinfo {author} {\bibfnamefont {A.~T.}\ \bibnamefont
  {Rezakhani}}, \bibinfo {author} {\bibfnamefont {A.~P.}\ \bibnamefont {Babu}},
  \bibinfo {author} {\bibfnamefont {K.}~\bibnamefont {M\o{}lmer}}, \bibinfo
  {author} {\bibfnamefont {M.}~\bibnamefont {M\"ott\"onen}},\ and\ \bibinfo
  {author} {\bibfnamefont {T.}~\bibnamefont {Ala-Nissila}},\ }\bibfield
  {title} {\bibinfo {title} {Correlation-picture approach to
  open-quantum-system dynamics},\ }\href
  {https://doi.org/10.1103/PhysRevX.10.041024} {\bibfield  {journal} {\bibinfo
  {journal} {Phys. Rev. X}\ }\textbf {\bibinfo {volume} {10}},\ \bibinfo
  {pages} {041024} (\bibinfo {year} {2020})}\BibitemShut {NoStop}%
\bibitem [{\citenamefont {Horowitz}\ and\ \citenamefont
  {Jarzynski}(2008)}]{Horowitz_Jarzynski_2008}%
  \BibitemOpen
  \bibfield  {author} {\bibinfo {author} {\bibfnamefont {J.}~\bibnamefont
  {Horowitz}}\ and\ \bibinfo {author} {\bibfnamefont {C.}~\bibnamefont
  {Jarzynski}},\ }\bibfield  {title} {\bibinfo {title} {Comment on “failure
  of the work-hamiltonian connection for free-energy calculations”},\ }\href
  {https://doi.org/10.1103/PhysRevLett.101.098901} {\bibfield  {journal}
  {\bibinfo  {journal} {Physical Review Letters}\ }\textbf {\bibinfo {volume}
  {101}},\ \bibinfo {pages} {098901} (\bibinfo {year} {2008})}\BibitemShut
  {NoStop}%
\bibitem [{\citenamefont {Peliti}(2008{\natexlab{a}})}]{Peliti_2008b}%
  \BibitemOpen
  \bibfield  {author} {\bibinfo {author} {\bibfnamefont {L.}~\bibnamefont
  {Peliti}},\ }\bibfield  {title} {\bibinfo {title} {Comment on “failure of
  the work-hamiltonian connection for free-energy calculations”},\ }\href
  {https://doi.org/10.1103/PhysRevLett.101.098903} {\bibfield  {journal}
  {\bibinfo  {journal} {Physical Review Letters}\ }\textbf {\bibinfo {volume}
  {101}},\ \bibinfo {pages} {098903} (\bibinfo {year}
  {2008}{\natexlab{a}})}\BibitemShut {NoStop}%
\bibitem [{\citenamefont {Vilar}\ and\ \citenamefont
  {Rubi}(2008)}]{Vilar_Rubi_2008b}%
  \BibitemOpen
  \bibfield  {author} {\bibinfo {author} {\bibfnamefont {J.~M.~G.}\
  \bibnamefont {Vilar}}\ and\ \bibinfo {author} {\bibfnamefont {J.~M.}\
  \bibnamefont {Rubi}},\ }\bibfield  {title} {\bibinfo {title} {Failure of the
  work-hamiltonian connection for free-energy calculations},\ }\href
  {https://doi.org/10.1103/PhysRevLett.100.020601} {\bibfield  {journal}
  {\bibinfo  {journal} {Physical Review Letters}\ }\textbf {\bibinfo {volume}
  {100}},\ \bibinfo {pages} {020601} (\bibinfo {year} {2008})}\BibitemShut
  {NoStop}%
\bibitem [{\citenamefont {Peliti}(2008{\natexlab{b}})}]{Peliti_2008a}%
  \BibitemOpen
  \bibfield  {author} {\bibinfo {author} {\bibfnamefont {L.}~\bibnamefont
  {Peliti}},\ }\bibfield  {title} {\bibinfo {title} {On the work–hamiltonian
  connection in manipulated systems},\ }\href
  {https://doi.org/10.1088/1742-5468/2008/05/P05002} {\bibfield  {journal}
  {\bibinfo  {journal} {Journal of Statistical Mechanics: Theory and
  Experiment}\ }\textbf {\bibinfo {volume} {2008}},\ \bibinfo {pages} {P05002}
  (\bibinfo {year} {2008}{\natexlab{b}})}\BibitemShut {NoStop}%
\bibitem [{\citenamefont {Hayden}\ and\ \citenamefont
  {Sorce}(2022)}]{Hayden_Sorce_2022}%
  \BibitemOpen
  \bibfield  {author} {\bibinfo {author} {\bibfnamefont {P.}~\bibnamefont
  {Hayden}}\ and\ \bibinfo {author} {\bibfnamefont {J.}~\bibnamefont {Sorce}},\
  }\bibfield  {title} {\bibinfo {title} {A canonical hamiltonian for open
  quantum systems},\ }\href {https://doi.org/10.1088/1751-8121/ac65c2}
  {\bibfield  {journal} {\bibinfo  {journal} {Journal of Physics A:
  Mathematical and Theoretical}\ }\textbf {\bibinfo {volume} {55}},\ \bibinfo
  {pages} {225302} (\bibinfo {year} {2022})}\BibitemShut {NoStop}%
\bibitem [{\citenamefont {Pal}\ \emph {et~al.}(2020)\citenamefont {Pal},
  \citenamefont {Saryal}, \citenamefont {Segal}, \citenamefont {Mahesh},\ and\
  \citenamefont {Agarwalla}}]{Pal_etal_2020}%
  \BibitemOpen
  \bibfield  {author} {\bibinfo {author} {\bibfnamefont {S.}~\bibnamefont
  {Pal}}, \bibinfo {author} {\bibfnamefont {S.}~\bibnamefont {Saryal}},
  \bibinfo {author} {\bibfnamefont {D.}~\bibnamefont {Segal}}, \bibinfo
  {author} {\bibfnamefont {T.~S.}\ \bibnamefont {Mahesh}},\ and\ \bibinfo
  {author} {\bibfnamefont {B.~K.}\ \bibnamefont {Agarwalla}},\ }\bibfield
  {title} {\bibinfo {title} {Experimental study of the thermodynamic
  uncertainty relation},\ }\href
  {https://doi.org/10.1103/PhysRevResearch.2.022044} {\bibfield  {journal}
  {\bibinfo  {journal} {Phys. Rev. Res.}\ }\textbf {\bibinfo {volume} {2}},\
  \bibinfo {pages} {022044} (\bibinfo {year} {2020})}\BibitemShut {NoStop}%
\bibitem [{\citenamefont {Micadei}\ \emph {et~al.}(2021)\citenamefont
  {Micadei}, \citenamefont {Peterson}, \citenamefont {Souza}, \citenamefont
  {Sarthour}, \citenamefont {Oliveira}, \citenamefont {Landi}, \citenamefont
  {Serra},\ and\ \citenamefont {Lutz}}]{Micadei_etal_2021}%
  \BibitemOpen
  \bibfield  {author} {\bibinfo {author} {\bibfnamefont {K.}~\bibnamefont
  {Micadei}}, \bibinfo {author} {\bibfnamefont {J.~P.~S.}\ \bibnamefont
  {Peterson}}, \bibinfo {author} {\bibfnamefont {A.~M.}\ \bibnamefont {Souza}},
  \bibinfo {author} {\bibfnamefont {R.~S.}\ \bibnamefont {Sarthour}}, \bibinfo
  {author} {\bibfnamefont {I.~S.}\ \bibnamefont {Oliveira}}, \bibinfo {author}
  {\bibfnamefont {G.~T.}\ \bibnamefont {Landi}}, \bibinfo {author}
  {\bibfnamefont {R.~M.}\ \bibnamefont {Serra}},\ and\ \bibinfo {author}
  {\bibfnamefont {E.}~\bibnamefont {Lutz}},\ }\bibfield  {title} {\bibinfo
  {title} {Experimental validation of fully quantum fluctuation theorems using
  dynamic bayesian networks},\ }\href
  {https://doi.org/10.1103/PhysRevLett.127.180603} {\bibfield  {journal}
  {\bibinfo  {journal} {Phys. Rev. Lett.}\ }\textbf {\bibinfo {volume} {127}},\
  \bibinfo {pages} {180603} (\bibinfo {year} {2021})}\BibitemShut {NoStop}%
\bibitem [{\citenamefont {Ashtekar}\ and\ \citenamefont
  {Schilling}(1999)}]{Ashtekar_1999}%
  \BibitemOpen
  \bibfield  {author} {\bibinfo {author} {\bibfnamefont {A.}~\bibnamefont
  {Ashtekar}}\ and\ \bibinfo {author} {\bibfnamefont {T.~A.}\ \bibnamefont
  {Schilling}},\ }\bibinfo {title} {Geometrical formulation of quantum
  mechanics},\ in\ \href {https://doi.org/10.1007/978-1-4612-1422-9_3} {\emph
  {\bibinfo {booktitle} {On Einstein's Path: Essays in Honor of Engelbert
  Schucking}}},\ \bibinfo {editor} {edited by\ \bibinfo {editor} {\bibfnamefont
  {A.}~\bibnamefont {Harvey}}}\ (\bibinfo  {publisher} {Springer New York},\
  \bibinfo {address} {New York, NY},\ \bibinfo {year} {1999})\ pp.\ \bibinfo
  {pages} {23--65}\BibitemShut {NoStop}%
\bibitem [{\citenamefont {Landau}\ \emph {et~al.}(2015)\citenamefont {Landau},
  \citenamefont {Páez},\ and\ \citenamefont
  {Bordeianu}}]{RH_Landau_etal_2015}%
  \BibitemOpen
  \bibfield  {author} {\bibinfo {author} {\bibfnamefont {R.~H.}\ \bibnamefont
  {Landau}}, \bibinfo {author} {\bibfnamefont {M.~J.}\ \bibnamefont {Páez}},\
  and\ \bibinfo {author} {\bibfnamefont {C.~C.}\ \bibnamefont {Bordeianu}},\
  }\href@noop {} {\emph {\bibinfo {title} {Computational physics: problem
  solving with Python}}},\ \bibinfo {edition} {3rd}\ ed.,\ Physics textbook\
  (\bibinfo  {publisher} {Wiley-VCH},\ \bibinfo {year} {2015})\BibitemShut
  {NoStop}%
\bibitem [{\citenamefont {Malavazi}\ and\ \citenamefont
  {Brito}(2022)}]{Malavazi_Brito_2022}%
  \BibitemOpen
  \bibfield  {author} {\bibinfo {author} {\bibfnamefont {A.~H.~A.}\
  \bibnamefont {Malavazi}}\ and\ \bibinfo {author} {\bibfnamefont
  {F.}~\bibnamefont {Brito}},\ }\bibfield  {title} {\bibinfo {title} {A schmidt
  decomposition approach to quantum thermodynamics},\ }\href
  {https://doi.org/10.3390/e24111645} {\bibfield  {journal} {\bibinfo
  {journal} {Entropy}\ }\textbf {\bibinfo {volume} {24}},\ \bibinfo {pages}
  {1645} (\bibinfo {year} {2022})}\BibitemShut {NoStop}%
\bibitem [{\citenamefont {Frey}\ \emph {et~al.}(2014)\citenamefont {Frey},
  \citenamefont {Funo},\ and\ \citenamefont {Hotta}}]{Frey_Funo_Hotta_2014}%
  \BibitemOpen
  \bibfield  {author} {\bibinfo {author} {\bibfnamefont {M.}~\bibnamefont
  {Frey}}, \bibinfo {author} {\bibfnamefont {K.}~\bibnamefont {Funo}},\ and\
  \bibinfo {author} {\bibfnamefont {M.}~\bibnamefont {Hotta}},\ }\bibfield
  {title} {\bibinfo {title} {Strong local passivity in finite quantum
  systems},\ }\href {https://doi.org/10.1103/PhysRevE.90.012127} {\bibfield
  {journal} {\bibinfo  {journal} {Phys. Rev. E}\ }\textbf {\bibinfo {volume}
  {90}},\ \bibinfo {pages} {012127} (\bibinfo {year} {2014})}\BibitemShut
  {NoStop}%
\bibitem [{\citenamefont {Alhambra}\ \emph {et~al.}(2019)\citenamefont
  {Alhambra}, \citenamefont {Styliaris}, \citenamefont {Rodriguez-Briones},
  \citenamefont {Sikora},\ and\ \citenamefont
  {Martín-Martínez}}]{Alhambra_etal_2019}%
  \BibitemOpen
  \bibfield  {author} {\bibinfo {author} {\bibfnamefont {A.~M.}\ \bibnamefont
  {Alhambra}}, \bibinfo {author} {\bibfnamefont {G.}~\bibnamefont {Styliaris}},
  \bibinfo {author} {\bibfnamefont {N.~A.}\ \bibnamefont {Rodriguez-Briones}},
  \bibinfo {author} {\bibfnamefont {J.}~\bibnamefont {Sikora}},\ and\ \bibinfo
  {author} {\bibfnamefont {E.}~\bibnamefont {Martín-Martínez}},\ }\bibfield
  {title} {\bibinfo {title} {Fundamental limitations to local energy extraction
  in quantum systems},\ }\href {https://doi.org/10.1103/PhysRevLett.123.190601}
  {\bibfield  {journal} {\bibinfo  {journal} {Physical Review Letters}\
  }\textbf {\bibinfo {volume} {123}},\ \bibinfo {pages} {190601} (\bibinfo
  {year} {2019})}\BibitemShut {NoStop}%
\end{thebibliography}%
